\documentclass{elsarticle}       
\usepackage{graphicx}
\usepackage{hyperref}
\usepackage{float}
\usepackage{amsmath,amssymb,amsthm,amsfonts}
\usepackage{lipsum}
\usepackage{graphicx}
\usepackage{epstopdf}
\usepackage{algorithm}
\usepackage{algorithmic}
\usepackage{multirow}
\usepackage{amsopn}
\usepackage[tight,footnotesize]{subfigure}
\usepackage{bigstrut}
\usepackage{hhline}
\usepackage{lscape}
\usepackage{ulem}
\usepackage{color}

\def\restrict#1{\raise-.5ex\hbox{\ensuremath|}_{#1}}
\newcommand{\IN}{\mathbb{N}}
\newcommand{\IR}{\mathbb{R}}
\newcommand{\m}[1]{\mathbf{#1}}

\DeclareMathOperator{\diag}{diag}

\DeclareMathOperator{\rank}{rank}

\DeclareMathOperator{\spn}{span}

\DeclareMathOperator{\CR}{CR}

\newdefinition{definition}{Definition}
\newtheorem{thm}{Theorem}
\newtheorem{proposition}{Proposition}
\newtheorem{lemma}{Lemma}
\newtheorem{corollary}{Corrolary}

\newcommand{\TheTitle}{Nonlinear least-squares spline fitting with variable knots}

\biboptions{longnamesfirst,angle,semicolon}

\begin{document}

\title{{\TheTitle}}

	\author[1,2]{P\'eter Kov\'acs\corref{cor1} } 
	\cortext[cor1]{Corresponding author. Tel.: +36 30 806 3000 / 8460; +43 732 2468 / 5686}
	\ead{kovika@inf.elte.hu, peter.kovacs@jku.at}
	\ead[url]{http://www.numanal.inf.elte.hu/\string~kovi/}
	\address[1]{Department of Numerical Analysis, E\"otv\"os L. University, 1117 Budapest, Hungary}
	\address[2]{Institue of Signal Processing, Johannes Kepler University Linz, 4040 Linz, Austria}
  \author[3]{Andrea M. Fekete}
	\ead{featact@inf.elte.hu}
	\address[3]{E\"otv\"os L. University, 1117 Budapest, Hungary}

\begin{abstract}
In this paper, we present a nonlinear least-squares fitting algorithm using B-splines with free knots. Since its performance strongly depends on the initial estimation of the free parameters (i.e. the knots), we also propose a fast and efficient knot-prediction algorithm that utilizes numerical properties of first-order B-splines. Using $\ell_p\;(p=1,2,\infty)$ norm solutions, we also provide three different strategies for properly selecting the free knots. Our initial predictions are then iteratively refined  by means of a gradient-based variable projection optimization. Our method is general in nature and can be used to estimate the optimal number of knots in cases in which no a-priori information is available.\\
To evaluate the performance of our method, we approximated a one-dimensional discrete time series and conducted an extensive comparative study using both synthetic and real-world data. We chose the problem of electrocardiogram (ECG) signal compression as a real-world case study. Our experiments on the well-known PhysioNet MIT-BIH Arrhythmia database show that the proposed method outperforms other knot-prediction techniques in terms of accuracy while requiring much lower computational complexity.

\end{abstract}


\maketitle

\section{Introduction}
\label{sec:intro}

While curve-fitting algorithms are used in many fields of applied sciences, this paper focuses on their signal processing aspects. Let us consider the normed space $(\mathcal{H},\left\|\cdot\right\|)$ of a specific class of real-valued signals over time. In this framework, the general nonlinear model of a particular $f\in\mathcal{H}$ can be given as follows:
\begin{equation}
\label{eq:nonlin_model}
f(t)\approx \eta(\m{c},\boldsymbol{\alpha};t)=\sum_{k=1}^n c_{k}\varphi_{k}(\boldsymbol{\alpha};t) \qquad (t\in\IR,\,\m{c}\in\IR^n,\,\boldsymbol{\alpha}\in\IR^m)\textrm{,}
\end{equation}
where the numbers of parameters $n,m\in\IN_+$ and the system of linearly independent functions $\left\{\varphi_k(\boldsymbol{\alpha};\cdot)\in\mathcal{H}\,|\,k=1,\ldots,n\right\}$ are predefined in accordance with the problem to be investigated. Usually, the Lebesgue spaces $L^p(\IR)$ for $p=1,2,\infty$ are chosen for $\mathcal{H}$ because they can be used in several contexts in signal processing, such as data-fitting, filter design, neural networks, classification, model reduction, and robotics. The best parameters can be defined as the solution to
\begin{equation}
\label{eq:lp_opt}
\underset{\m{c}\in\IR^n,\,\boldsymbol{\alpha}\in\IR^m}{\min} r_p(\mathbf{c},\boldsymbol{\alpha})= \underset{\m{c}\in\IR^n,\,\boldsymbol{\alpha}\in\IR^m}{\min} \Big\| f(\cdot) - \eta(\m{c},\boldsymbol{\alpha};\cdot) \Big\|_p \qquad (f\in L^p(\IR))\textrm{.}
\end{equation}
This problem can be simplified for $p=2$, since $L^2(\IR)$ is a Hilbert space with regard to the usual dot product $\left\langle\cdot,\cdot \right\rangle$ and the corresponding induced norm $\left\|f\right\|_2=\sqrt{\left\langle f, f \right\rangle}$. In this case, the parameters are separable, and for a given $\boldsymbol{\alpha}$, the optimal value of $\m{c}$ can be computed explicitly via the orthogonal projection $P_{\mathcal{S}(\boldsymbol{\alpha})}$ onto the subspace $\mathcal{S}(\boldsymbol{\alpha}):=\spn\left\{\varphi_{k}(\boldsymbol{\alpha};\cdot)\in L^2(\IR)\,|\, k=1,\ldots,n \right\}$. The original problem in Eq.~\eqref{eq:lp_opt} can therefore be reduced to 
\begin{equation}
	\min_{\boldsymbol{\alpha}\in\IR^m} r_2(\boldsymbol{\alpha})= \min_{\boldsymbol{\alpha}\in\IR^m} \left\|f-P_{\mathcal{S}(\boldsymbol{\alpha})}f\right\|^2_2= \min_{\boldsymbol{\alpha}\in\IR^m} 	\big\Vert P^{\bot}_{\mathcal{S}(\boldsymbol{\alpha})}f\big\Vert^2_2\,,
\label{eq:varpro_min}
\end{equation}
where $r_2$ is called the \textsl{variable projection functional} (VP), and $P^{\bot}_{\mathcal{S}(\boldsymbol{\alpha})}$ is the projector on the orthogonal complement of $\mathcal{S}(\boldsymbol{\alpha})$ with respect to $L^2(\IR)$. In practice, the discrete variant of the problem is considered, that is, $\mathcal{H}=\IR^N,\;(N\in\IN_+),$ and the $\ell_2$ norm is used. This special case is a separable nonlinear least-squares problem, which has been investigated by many authors.

In his 1970 article, Scolnik \cite{scolnik1970} described the solution to Eq.~\eqref{eq:varpro_min} for the case of exponential type base functions, which was later extended by Guttman et.\ al.~\cite{gps1971} to general functions with one variable. Lawton and Sylvestre \cite{lawton-sylvestre} gave a numerical solution to the same problem with several variables, approximating the derivatives of $r_2$ by finite differences. Later, Golub and Pereyra \cite{golub-pereyra} constructed the exact formula for the Fr\'echet derivative of $r_2$ with respect to $\boldsymbol{\alpha}$. An extensive review of the related results along with various applications of nonlinear least-squares can be found in \cite{golub_pereyra2003}. In order to generalize the original VP problem by incorporating Tikhonov regularization, Chung and Nagy \cite{tikhonov_varpro} developed a hybrid iterative approach in which the regularization parameter can be chosen automatically. Cornelio et al.~\cite{nonegative_varpro} adapted this approach for blind deconvolution problems, including nonnegativity constraints on the nonlinear parameters. A recent publication by O'Leary and Rust \cite{varpro_matlab} has summarized the evolution of VP implementations in various programming languages, such as FORTRAN, R and \textsc{MatLab}. They also developed a \textsc{MatLab} implementation for optimizing the VP functional that allows constraints and weights to be used.

Piecewise polynomial approximations (e.g.\ splines) play a substantial role in data-fitting. These are flexible curves that can easily be parametrized via knots. Substituting the elementary waves $\varphi_k(\boldsymbol{\alpha};\cdot)$ with B-splines requires the optimal vector of knots $\boldsymbol{\alpha}\in\IR^m$ in Eq.~\eqref{eq:varpro_min} to be determined. The resulting optimization problem is difficult to solve since the VP functional is nonlinear and nonconvex. As shown by Jupp, the main problem is caused by identical knots, which induce numerous stationary points (see, e.g.,\ the Lethargy Theorem in \cite{jupp}). To avoid coalescing knots, Jupp \cite{jupp} proposed a logarithmic transformation for pushing the knot set boundaries to infinity. Penalizing the distance between knots is a similar approach used in \cite{dierckx, guertin, lindstrom}. Other authors, such as Molinari et al.~\cite{bdoptknot}, defined constraints on the knot sequence, while Beliakov \cite{beliakov} utilized global optimization techniques to solve the problem. Borges and Pastva~\cite{varpro_bezier} subsequently reformulated the VP problem for nonlinear B\'ezier curve fitting, making it applicable to computer-aided design. Since, from a signal processing point of view, the domain of possible knot sequences can be considered a discrete set (i.e.\ the sample indices), combinatorial optimization techniques are also applicable \cite{yoshimoto}. The most commonly used heuristic is greedy backward deletion of knots. A review of these algorithms and an application for compressing electrocardiograms (ECG) can be found in \cite{bdoptknot} and \cite{ekg}, respectively. 

The aforementioned procedures and implementations continue to have some shortcomings. Some cannot reliably find a global optimum or suffer from the consequences of the Lethargy Effect, while others that overcome these problems require time-intensive computations. We developed an algorithm that addresses these problems, requires significantly less execution time, and can be used to initialize other (i.e.\ gradient-based) methods. We constructed our algorithm using first-order B-splines. These are simple piecewise constant functions that have the advantage of expressing the error in an analytic form. We can predict the positions of the knots by solving Eq.~\eqref{eq:varpro_min} for first-order B-spline approximations (FOBA). In this special case, we provide three different strategies for finding knots by considering not just $\ell_2$, but also $\ell_1, \, \ell_\infty$ norm solutions. The resulting knots can be used to initialize other optimization methods to find better fitting splines of higher order. \par

In order to demonstrate the efficiency of our method in a real-world application, we chose a task from the field of signal compression: decreasing the size of data while keeping the reconstruction error as low as possible. Dimensional reduction of a signal is often a preparatory step before other methods, such as machine learning, are employed. We tested our algorithm with ECG recordings because the signals are especially long (some medical procedures, such as Holter monitoring, can require up to 24 hours of multi-channel recordings). The test set comprised 11 hours of data and included real measurement noise. Our method proved to be robust and suitable for real-time processing of ECGs. The compressed data (i.e.\ knots and coefficients) could also be used in further processing steps, for instance, to analyze irregularities in heartbeats. 



\section{Background}
\label{sec:problem}

\subsection{B-splines}
Let us consider the interval $[a,b]\subset\IR$ and the sequence of ordered knots $a=t_0< t_1< \ldots < t_{n-1} < t_n=b$. We denote the set of polynomials of degree at most $\ell$ by $\mathcal{P}_{\ell}$, and the collection of $\ell$-times continuously differentiable functions on $[a,b]$ by $\mathcal{C}^{\ell}[a,b]$.
 
\begin{definition} $s:[a,b]\to\mathbb{R}$ is a \textsl{spline} of order $\ell+1$ (or degree $\ell$) if
\label{def:splines}
\begin{enumerate}
\renewcommand\labelenumi{(\roman{enumi})}
\item $s\restrict{[t_{k},t_{k+1}]}\in\,\mathcal{P}_{\ell}  \;\; (k=0,\ldots ,n-1)$,
\item $s \in \mathcal{C}^{\ell-1}[a,b]$.
\end{enumerate}
\end{definition}
We will use the vector of knots $\tau_n=\left(t_{-\ell},t_{-\ell+1},\ldots ,t_{n+\ell}\right)^T$ $(t_k\in[a,b])$ with the following boundary conditions:
\begin{equation}
	t_{-\ell}=t_{-\ell+1}=\ldots =t_0=a \quad \text{ and } \quad t_n=t_{n+1}=\ldots =t_{n+\ell}=b\,\textrm{.}
\label{eq:bound_eq}
\end{equation}
Let $\mathcal{S}_{\ell}(\tau_n)$ stand for the corresponding set of spline functions of degree $\ell$. First-order B-splines ($\ell=0$) can then be defined as follows:
\begin{equation}
\label{eq:firstorder}
N_{0,k}(\tau_n;x):=B_{0,k}(\tau_n;x):=\chi_{[t_{k}, t_{k+1})}(x)=
\begin{cases}
1 & \textrm{if $x\in[t_{k}, t_{k+1})$}\textrm{,} \\
0 & \textrm{otherwise}\textrm{.}   
\end{cases}
\end{equation}
Higher-order B-splines ($\ell \geq 1$) are given by the recursion
\begin{equation}
N_{\ell,k}(\tau_n;x)=C \cdot B_{\ell,k}(\tau_n;x)=\left((-1)^{\ell +1}(t_{k+\ell +1}-t_k)\right) \cdot \left[ t_k, \ldots , t_{k+\ell +1}\right](x-t)_+^\ell\,\textrm{,}
\label{eq:bspline_def}
\end{equation}
where $(x-t)^{\ell}_+=(\max \left\{x-t,0\right\})^{\ell}$ are the so-called \textsl{truncated power functions} (TPF). In this definition, the first variable $x$ of each TPF is fixed, and the $(\ell+1)$th divided differences $\left[ t_k, \ldots , t_{k+\ell +1}\right](x-t)_+^\ell$ are computed for the second variable $t$. Note that the usual definition of B-splines would be $N_{\ell,k}$, but for the sake of simplicity we ignore the scaling factor $C$ and apply the notation $B_{\ell,k}$. This form of the B-spline functions $B_{\ell,k}\;(\ell\geq 1)$ yields the following formula for the partial derivatives with respect to the free knots:
\begin{equation}
\frac{\partial B_{\ell,j}(\tau_n;x)}{\partial t_k}=
    \begin{cases}
      \left[ t_j, \ldots t_k,t_k \ldots , t_{j+\ell +1}\right](x-t)_+^\ell & \text{if}\ j\leq k \leq j+\ell+1\textrm{,}\\
      0 & \text{otherwise}\textrm{.}
    \end{cases}
		\label{eq:derBspline}
\end{equation} 
According to Curry and Schoenberg \cite{curryschoenberg}, the linear space $\mathcal{S}_{\ell}(\tau_n)$ is spanned by the B-spline functions $B_{\ell,k}\; (k=-\ell,\ldots,n-1)$. These functions are linearly independent, and thus $\dim \mathcal{S}_{\ell}(\tau_n)=n+\ell$ provided that $t_0,\ldots,t_n$ are pairwise different knots. Guaranteeing this property for the discrete analogue of the problem requires the Schoenberg--Whitney condition to be satisfied: each B-spline's support should contain at least one sample point \cite{schoenberg}.

	In real-world applications, shorter segments of the complete signal are processed. The outer knots can therefore assumed to be fixed, while the inner points $t_1,\ldots ,t_{n-1}\in[a,b]$ are variable. Let $\boldsymbol{\alpha}\in\IR^{n-1}$, and denote the vector of variable knots with boundary conditions by $\tau_n(\boldsymbol{\alpha}):=(t_{-\ell},\ldots,t_0,\boldsymbol{\alpha},t_n,\ldots,t_{n+\ell})^T$. We can now formalize a special case of the VP problem in Eq~\eqref{eq:varpro_min} using B-splines: For a given $\ell\in\IN_+$, the subspace $\mathcal{S}(\boldsymbol{\alpha})=\mathcal{S}_{\ell}(\tau_n(\boldsymbol{\alpha}))$ and the $\varphi_k$ base functions are the corresponding B-splines of degree $\ell$.
	
\subsection{Lethargy Effect}
Divided differences, and thus B-splines, are symmetric functions with respect to their arguments. This implies that if $\boldsymbol{\alpha}\in\IR^{n-1}$ is either a local or a global extremum of the B-spline VP problem in Eq.~\eqref{eq:varpro_min}, any of its permutations will also satisfy this criterion. Another difficulty was discussed by Jupp in \cite{jupp_gammapol,jupp}. As described in \cite{jupp}, the problem can be inspected by first introducing the set:
\begin{equation*}
	s_{n-1}[a,b]=\left\{\boldsymbol{\alpha}\in\IR^{n-1}\,:\,a=t_0<\alpha_1<\alpha_2<\ldots<\alpha_{n-1}<t_n=b\right\}.
\end{equation*}
The $\overline{s}_{n-1}$ closure of said $s_{n-1}$ is an $(n-1)$-simplex that is given by the following constraints:
\begin{equation}
\label{eq:juppconstraint}
	(\alpha_p-\alpha_{p-1}) \geq 0 \qquad (p=1,2,\ldots,n)\,\textrm{.}
\end{equation}
Finally, let $s_{n-1}^{(p)}$ be defined as the \textsl{$p$th (open) main face} of  $\overline{s}_{n-1}$ for which only the $p$th constraint is active (i.e.\ $\alpha_p=\alpha_{p-1}$). On each of these main faces, the B-spline VP functional $r_2(\boldsymbol{\alpha})$ is symmetrical with respect to interchanging the variables $\alpha_{p-1}$ and $\alpha_p$. 
\begin{thm}[Jupp \cite{jupp}, ``Lethargy Theorem'']
\label{theorem:jupp}
On the $p$th main face, $s_{n-1}^{(p)}$,
\begin{equation*}
	\mathbf{n}_p^T \nabla r_2(\boldsymbol{\alpha})=0 \quad  (p=2,3,\ldots,n-1)\,,
\end{equation*}
where $\mathbf{n}_p$ is the unit outward normal to $s_{n-1}^{(p)}$.
\end{thm} 	

Note that the statement of this theorem is independent of the approximated function $f$. As a consequence, the main faces contain many \textsl{stationary points} (i.e.\ extrema or saddle points), at which the gradient is zero. This can cause gradient-based methods to stop prematurely. For example, if two knots get too close to one another, the gradient in the $\mathbf{n}_p$ direction is small, and the algorithm will therefore erroneously search for the optimum on the $p$th main face. Fig.~\ref{fig:pelda} shows this phenomenon for cubic splines ($\ell=3$), where $\boldsymbol{\alpha}\in\IR^2$. The upper graph in Fig.~\ref{fig:face} plots the cross section of $r_2$ along the main face. Here, the knot vectors $\boldsymbol{\alpha}^{(1)},\boldsymbol{\alpha}^{(3)}$ are local minima, while $\boldsymbol{\alpha}^{(2)}$ is a saddle point of $r_2$. There are two global minima (green crosses), for which the corresponding cubic B-spline approximation is shown in the lower graph in Fig.~\ref{fig:face}. The problems mentioned above are particularly important in the context of ECG recordings. These signals contain diagnostically important waves (e.g. the QRS complex), which require more knots to be inserted in a small area. Since these knots will be relatively close to each other, the problems resulting from the Lethargy Theorem can in some cases manifest in practice.

\begin{figure}[!htb]
\centering
  \subfigure[Graph of $r_2(\boldsymbol{\alpha})$ where $\boldsymbol{\alpha}\in\IR^2$.]{
  \includegraphics[scale=0.50, trim=140 265 154 265, clip]{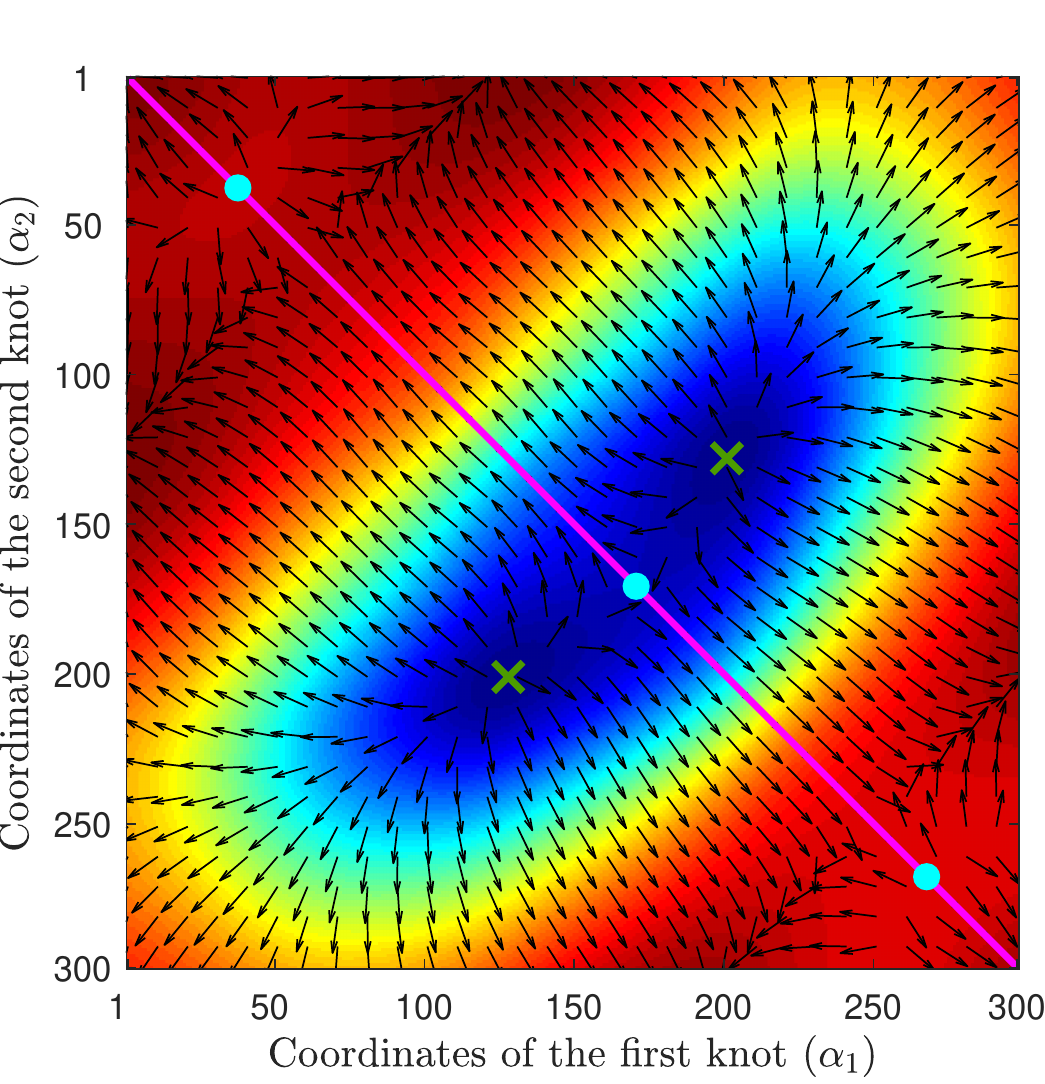}
  \label{fig:varpro2d}
  }\hspace{2mm}
  \subfigure[Extrema of the $s_{2}^{(2)}$ main face.]{
  \includegraphics[scale=0.50, trim=140 265 135 265, clip]{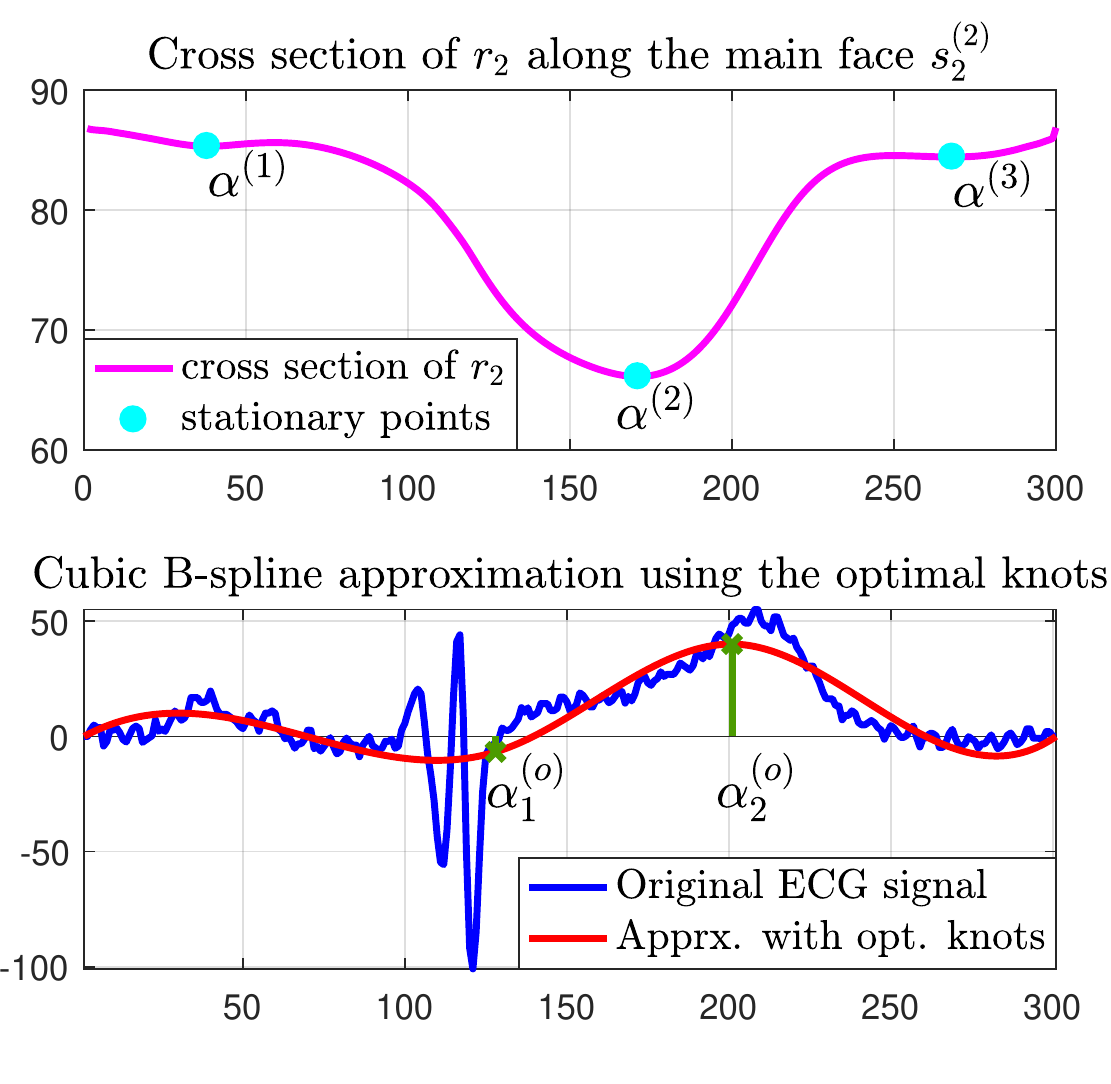}
  \label{fig:face}
  } 
\caption{Critical points of $r_2(\boldsymbol{\alpha})$ for two free knots. The black arrows indicate the normalized gradient vectors of the objective function $r_2(\boldsymbol{\alpha})$. The diagonal line represents the second main face, and its stationary points are marked by light blue dots.}
\label{fig:pelda}
\end{figure}

\section{First-order B-spline approximations (FOBA)}
\label{sec:foba}

\subsection{\texorpdfstring{Solution for $p=2$}{Solution for $p=2$}}

Finding an exact formula for the minimum of the B-spline VP functional is almost impossible, since changing one knot affects both the base functions and the corresponding coefficients in a nonlinear way. The only case in which this could be done is $\ell=0$, where the first-order B-splines are equal to the characteristic functions of subintervals $[t_k, t_{k+1})$. As the supports of these functions are distinct, changing the $q$th knot affects only two base functions: $B_{0,q-1}$ and $B_{0,q}$ and their corresponding coefficients $c_{q-1}$ and $c_q$. We use this simple relation to develop an error formula and to predict the positions of the knots. 

Provided that the functions $\varphi_k(\boldsymbol{\alpha},\cdot)\;(k=1,\ldots,n)$ are linearly independent for any value of $\boldsymbol{\alpha}$, the least-squares error in Eq.~\eqref{eq:varpro_min} can be written as
\begin{equation}
\| f-P_{\mathcal{S}(\boldsymbol{\alpha})}f\|_2^2=\|f\|_2^2-\m{b}^T(\boldsymbol{\alpha}) \m{c}(\boldsymbol{\alpha})=\|f\|_2^2-\m{b}^T(\boldsymbol{\alpha}) \m{G}^{-1}(\boldsymbol{\alpha}) \mathbf{b}(\boldsymbol{\alpha})\,\textrm{,}
\label{eq:nonorthog_bessel}
\end{equation}
where  $\m{G}_{ij}(\boldsymbol{\alpha})=\left\langle \varphi_i(\boldsymbol{\alpha};\cdot),\varphi_j(\boldsymbol{\alpha};\cdot)\right\rangle$, $\m{b}_i(\boldsymbol{\alpha})=\left\langle f(\cdot),\varphi_i(\boldsymbol{\alpha};\cdot)\right\rangle$ for $1\leq i,j\leq n$, and $\mathbf{b}^T(\boldsymbol{\alpha}) \m{G}^{-1}(\boldsymbol{\alpha}) \mathbf{b}(\boldsymbol{\alpha})>0$. Therefore, solving Eq.~\eqref{eq:varpro_min} is equivalent to 
\begin{equation}
	\max_{\boldsymbol{\alpha}\in\IR^m}\,\m{b}^T(\boldsymbol{\alpha}) \m{G}^{-1}(\boldsymbol{\alpha})\m{b}(\boldsymbol{\alpha})\textrm{.}
	\label{eq:varpro_max}
\end{equation}

Due to its orthogonal property, a system consisting of first-order B-splines $B_{0,k}\;(0\leq k\leq n-1)$ reduces the complexity of computations. Thus, for the vector of free knots $\boldsymbol{\alpha}\in\IR^{n-1}$, the corresponding Gramian matrix $\m{G}(\boldsymbol{\alpha})$ is diagonal:
\small
\begin{equation*}
\m{G}(\boldsymbol{\alpha})=\diag\Big(\langle B_{0,0},B_{0,0}\rangle,\ldots ,\langle B_{0,n-1},B_{0,n-1}\rangle\Big)=\diag\Big(\| B_{0,0}\|_2^2,\ldots ,\| B_{0,n-1}\|_2^2\Big)\textrm{,}
\end{equation*}
\normalsize
where the functions $B_{0,k}$ are defined by the knot vector $\tau_n(\boldsymbol{\alpha})$. The squared norms of these B-splines on the interval $[a,b]$ can easily be calculated as
\begin{equation*}
\| B_{0,k}\|_2^2=\langle B_{0,k},B_{0,k}\rangle=\int_{a}^{b} \chi^2_{[t_{k}, t_{k+1})}(x)  \, \mathrm{d}x=t_{k+1}-t_k\qquad (0\leq k \leq n-1)\textrm{.}
\end{equation*}
Similarly, for the $k$th element of the vector $\m{b}(\boldsymbol{\alpha})$ applies the following:
\begin{equation*}
\m{b}_k(\boldsymbol{\alpha})=\langle f,B_{0,k}\rangle=\int_{a}^{b} f(x)\cdot \chi_{[t_{k}, t_{k+1})}(x) \,\mathrm{d}x=\int_{t_{k}}^{t_{k+1}} f(x)\,\mathrm{d}x=:F_k(t_{k+1})
\end{equation*}
with $0\leq k \leq n-1$. In summary, the linear system of equations to be solved is
\small
\begin{equation*}
\left[
\begin{matrix}
\|B_{0,0}\|_2^2&0&\ldots&\ldots&\ldots&0\\
0&\ddots&&&&\vdots&\\
\vdots&&\|B_{0,q-1}\|_2^2&&&\vdots\\
\vdots&&&\|B_{0,q}\|_2^2&&\vdots\\
\vdots&&&&\ddots&0\\
0&\ldots&\ldots&\ldots&0&\|B_{0,n-1}\|_2^2\\
\end{matrix}
\right]
\left[
\begin{array}{l}
c_{0}\\
\vdots\\
c_{q-1}\\
c_{q}\\
\vdots\\
c_{n-1}
\end{array}
\right]
\hspace{-1mm}
=\hspace{-1mm}
\left[
\begin{array}{l}
F_{0}(t_1)\\
\vdots\\
F_{q-1}(t_q)\\
F_q(t_{q+1})\\
\vdots\\
F_{n-1}(t_n)
\end{array}
\right]\textrm{.}
\end{equation*}
\normalsize
Let us consider the case in which all knots are fixed except $t_q\in (t_{q-1},t_{q+1})$, which is free and whose optimal position is to be determined so that it maximizes Eq.~\eqref{eq:varpro_max}. In the case of first-order B-splines, moving the $q$th knot in $(t_{q-1},t_{q+1})$ changes only the coefficients $c_{q-1}$ and $c_q$, for which we have the formulas
\begin{equation}
c_{q-1}=\frac{F_{q-1}(t_q)}{\|B_{0,q-1}\|_2^2}=\frac{F_{q-1}(t_q)}{t_q-t_{q-1}}\textrm{,}\qquad c_{q}=\frac{F_{q}(t_{q+1})}{\|B_{0,q}\|_2^2}=\frac{F_{q}(t_{q+1})}{t_{q+1}-t_{q}}\textrm{.}
\label{eq:l2coeffs}
\end{equation}
In order to decide where the new knot is to be inserted, we compute the optimal positions $\alpha_q$ for all the intervals $[t_q,t_{q+1}]$ and then choose the best among them. That is for each $q=0,\ldots,n-1$, the following maximum search problem must be solved:
\begin{equation*}
\alpha_q=\underset{t_{q}\leq \alpha \leq t_{q+1}}{\arg\max}\,\m{b}^T(\alpha)\m{G}^{-1}(\alpha)\m{b}(\alpha)=\underset{t_{q} \leq \alpha \leq t_{q+1}}{\arg\max}\, \frac{\left(F_{q}(\alpha)\right)^2}{\alpha-t_{q}}+\frac{\left(F_{q+1}(\alpha)\right)^2}{t_{q+1}-\alpha}\textrm{.}
\end{equation*}
The problem can be further simplified because $F_{q+1}(\alpha)$ is equal to the difference between $F_{q}(t_{q+1})$ and $F_{q}(\alpha)$. We can also change the sign of the cost function above to obtain a minimum search problem, and the following proposition:
\begin{proposition}
\label{theorem:opt_e2}
Let us consider the knots $a=t_0<t_1<\ldots<t_n=b$, the corresponding system of first order B-splines $B_{0,k}\;(0\leq k\leq n-1)$, and the function $f\in\mathcal{C}[a,b]$. Inserting a new single knot $\alpha\in[a,b]$ according to Eq.~\eqref{eq:varpro_min} is then equivalent to solving the optimization problem: 
\begin{equation}
\label{eq:maxoshiba}
\underset{{0\leq q<n}}{\min} \; \underset{t_{q} \leq \alpha \leq t_{q+1}}{\min}\;e_2(\alpha):=\underset{{0\leq q < n}}{\min} \; \underset{t_{q} \leq \alpha \leq t_{q+1}}{\min}\; -\frac{\left(F_{q}(\alpha)\right)^2}{\alpha-t_{q}}-\frac{\left(F_{q}(t_{q+1}) - F_{q}(\alpha)\right)^2}{t_{q+1}-\alpha}\,\mathrm{.}
\end{equation}
\end{proposition}
Note that the functions $F_{q}$ and $F_{q+1}$, and thus also $e_2$, are differentiable provided that $f\in\mathcal{C}[a,b]$. Additionally, the proper one-sided limits are finite at the boundary points $t_{q},\,t_{q+1}$. Hence, the function $e_2$ can be extended continuously on the whole interval $\left[t_{q},t_{q+1}\right]$, which means that there exists at least one global minimum. It follows that gradient-based optimizations can be applied to find the points $\alpha_q$ that minimize $e_2$ for each interval $[t_q,t_{q+1}]$. 

\subsection{\texorpdfstring{Solution for $p=1,\infty$}{Solution for $p=1,\infty$}}
\label{sec:foba_othernorms}

The general problem in Eq.~\eqref{eq:lp_opt} becomes more difficult for $p=1,\infty$. Although, the minimum of the full functional still exists for a given $\boldsymbol{\alpha}\in\IR^m$, computing the corresponding coefficient vector $\m{c}\in\IR^n$ is a time-consuming iterative process \cite{cadzow,hegedus}. Addressing this issue, we reuse the idea from the previous section and restrict the optimization to a single knot. We assume that a certain number of knots $a=t_0 < t_1 <\ldots < t_{n-1} < t_n=b$ have already been inserted into the knot vector and that the position of the next knot is to be determined. For a system of first-order B-splines, changing a particular knot affects only two consecutive base functions and their corresponding coefficients. Thus, for $p=1,\infty$, the optimization in Eq.~\eqref{eq:lp_opt} is simplified to the following subproblems:
\begin{equation}
\underset{\m{c}\in\IR^2,\, t_q\leq \alpha \leq t_{q+1}}{\min} r_p(\m{c},\alpha)\qquad (0\leq q<n)\,.
\label{eq:lp_opt_singleknot}
\end{equation}
First-order B-splines are piecewise constant functions for which the corresponding coefficients are well defined in these norms \cite{natanson,watson}. Let us consider the function $f\in\mathcal{C}[a,b]$, which we want to approximate in the form $c^{(p)}\cdot \chi_{[a,b]}$. The coefficients of the best linear approximations in the cases $p=1,\infty$ are then
\begin{equation}
	c^{(1)}=\text{med} f , \qquad c^{(\infty)}=\frac{\min f+\max f}{2},
\label{eq:recall_lpcoeff}
\end{equation}
where $\text{med} f,\,\min f$ and $\max f$ denote the median, the minimum and the maximum values of $f$, respectively. We are now able to find the optimal knot candidate $\alpha_q$ in each interval $[t_q,t_{q+1}]$ for $q=0,\ldots,n-1$, and the final choice is the best among them.  
\begin{proposition}
\label{theorem:opt_e1inf}
In the notations of Eq.~\eqref{eq:firstorder} and under the assumptions of Proposition~\ref{theorem:opt_e2}, the position of a new single knot, according to Eq.~\eqref{eq:lp_opt_singleknot} with $p=1,\infty$, is given by the optimization problem
\small
\begin{equation*}
\underset{{0\leq q<n}}{\min} \; \underset{t_q\leq\alpha\leq t_{q+1}}{\min} e_p(\alpha)=\underset{{0\leq q<n}}{\min} \;\underset{t_q\leq\alpha\leq t_{q+1}}{\min} \left\|f\restrict{[t_{q},t_{q+1}]}- c^{(p)}_1\cdot\chi_{[t_q,\alpha]} - c^{(p)}_2\cdot \chi_{[\alpha,t_{q+1}]} \right\|_p\textrm{.}
\label{eq:lp_opt_efun}
\end{equation*}
\normalsize
\end{proposition}

\section{Optimization of the free knots}

\subsection{Knot-prediction algorithm}
\label{sec:knotpred}

We use the ideas proposed in the previous section to predict the optimal knots of first-order B-spline approximations in the sense of Eq.~\eqref{eq:lp_opt}. The resulting knot vector can be used to initialize numerical optimization methods to find better-fitting splines of higher order. Note that the same approach can be used for the three cases $p=1,2,\infty$. In fact, only the objective function $e_p$ should be changed in the algorithm.

\begin{algorithm}
\caption{Knot-prediction by FOBA.}
\label{alg:knotpred}
\begin{algorithmic}[1]
\STATE \textbf{function} \textsc{KnotPred}($a,\,b,\,f,\,n,\,\delta,\,p$)
\STATE Initialize the knot vector $\m{t}=(a,\,b)^T$	
\STATE Initialize the vector of candidate knots $\boldsymbol{\alpha}$ and their errors $\m{e}$ 
\STATE{Insert the first knot into $\m{t}$ and set $q=1$}
\WHILE{$\dim \m{t} \leq n$}
\FOR{$j=q,\,q+1$}
\IF{$(t_{j+1}-t_j)>\delta$}
\STATE{Find the optimal knot $\alpha_j\in[t_j+\delta,\,t_{j+1}-\delta]$}
\STATE{Insert $\alpha_j$ into the vector of candidate knots $\boldsymbol{\alpha}$}
\STATE{Insert $e_p(\alpha_j)$ into the vector of candidate knot errors $\m{e}$}
\ENDIF
\ENDFOR
\STATE{Update the other elements of $\m{e}$ with the values of $e(\alpha_q)$ and $e(\alpha_{q+1})$}
\STATE{Choose the leftmost candidate knot $\alpha_{opt}$ with the smallest value in $\m{e}$}
\STATE{Insert $\alpha_{opt}$ into the knot vector $\m{t}$}
\STATE{Set $q$ to the corresponding index of $\alpha_{opt}$ in $\m{t}$}
\ENDWHILE
\RETURN $\m{t}$
\STATE \textbf{end function}
\end{algorithmic}
\end{algorithm}
An outline of the proposed method is provided in Alg.~\ref{alg:knotpred}, where $f\in L^p([a,b])$, the number of knots is $n+1$, and the norm index $p$ is given by the user, while the positions of the predicted knots $\m{t}$ are returned. The parameter $\delta \geq 0$ controls the minimum distance between consecutive knots; that is, if two neighboring knots $t_j,t_{j+1}$ are too close to each other, we do not insert another knot in between them (e.g., \ $\delta=1$ for discrete time series). Note that the error values must be updated, so that they correspond to the improvement gained by inserting the $q$th knot in the previous iteration. This step can be found in line $13$ of Alg.~\ref{alg:knotpred}. Furthermore, if the vector $\m{e}$ has more than one minimum, then let us choose the smallest one possible (i.e.,\ the leftmost) as the candidate knot $\alpha_{opt}$. 

The first-order B-splines $B_{0,k}\;(k=0,\ldots,n-1)$ intersect the fitted function $f$ in at least one point over their support $[t_k,t_{k+1}]$. This property is independent of the norm $(p=1,2,\infty)$ used in the approximation. Additionally, when $f\in\mathcal{C}^1[a,b]$, we can apply the well-known error formula for polynomial interpolation 
\begin{equation*}
	\left|f(x)-B_{0,k}(\tau_n;x)\right| \leq M_k \cdot (t_{k+1}-t_{k})=M_k\cdot h_k \qquad \left(x\in[t_k,t_{k+1}]\right)\,,
\end{equation*}
where $M_k=\max_{x\in[t_k,t_{k+1}]}\left|f'(x)\right|$. This estimation reveals that the first derivative $M_k$ is inversely proportional to the knot spacing $h_k$ (see Lectures 10-11 in \cite{stewart}). Therefore, the number of knots should be high near high derivatives of $f$. This is why Alg.~\ref{alg:knotpred} inserts numerous knots near sudden changes of the signal. 

On one hand, the method described above uses a greedy strategy and optimizes only a single coordinate of $\boldsymbol{\alpha}\in\IR^{n-1}$ in each step. Consequently, it provides a suboptimal solution to the full problem in Eq.~\eqref{eq:lp_opt}. On the other hand, it is possible to predict the exact positions of the knots in simple cases, when $f$ is a step function. 
\begin{lemma}
\label{theorem:foba}
Let $\ell=0$, $\tau_{n}$ be the original knot vector and $f\in\mathcal{S}_0(\tau_{n})$ be a step function for which the number of knots $m$, $m < n$, is to be predicted. Alg.~\ref{alg:knotpred} with  $p=1,2,\infty$ then returns a subset of the original knots.
\end{lemma}
\begin{proof} It is sufficient to show that the proposed method predicts the exact position of a specific knot in each iteration. The proof can be easily constructed by induction for all the three cases $p=1,2,\infty$. 
\end{proof}

\begin{corollary}
\label{theorem:foba_on_splines}
Let $\ell > 0$ and the original knot vector $\tau_{n}$ be defined with boundary conditions in Eq.~\eqref{eq:bound_eq}, and $f\in\mathcal{S}_{\ell}(\tau_{n})$. Alg.~\ref{alg:knotpred} with  $p=1,2,\infty$ can then predict a certain number of knots $m$, $m < n$, accurately from $\tau_{n}$.
\end{corollary}
\begin{proof} 
According to Definition~\ref{def:splines}, the $\ell$th derivative of $f$ is a piecewise constant function $f^{(\ell)}\in S_0(\tau_{n})$ to which Lemma~\ref{theorem:foba} applies. 
\end{proof}
Note that, if $f^{(\ell)}$ retains the same value over the interval $[{t}_q,{t}_{q+2}]$, then the ``hidden knot'' ${t}_{q+1}$ cannot be localized by Alg.~\ref{alg:knotpred}. However, it is possible to find every knot of a spline $f\in\mathcal{S}_{\ell}({\tau}_{n})$ provided that all the consecutive steps of the function $f^{(\ell)}$ are different. 
In Section~\ref{sec:experiments}, we show that, according to experiments, the proposed method can predict the optimal knots not only for splines, but also for more complex functions.

The nature of the knot-prediction algorithm depends on the norm, which should be chosen to suit the given task. Since the predictions for $p=1$ are less sensitive to outliers than those for $p=2$, the former is preferred in the case of noisy data. This is due to the coefficients of first-order B-spline approximations being defined by the medians and the means of the data  for $p=1,2$, respectively (see Eqs.~\eqref{eq:l2coeffs}-\eqref{eq:recall_lpcoeff}). In the case of uniform approximations, the largest error magnitude is to be minimized. This property is useful for detecting sudden changes in the signal, such as spikes. These phenomena can be seen in Fig.~\ref{fig:foba_examples}, where we predicted the knots of a heartbeat for a real ECG signal (blue curve). For $p=\infty$, more than half of the interior knots are inserted near the middle lobe, which is called the QRS complex and one of the most important waveforms. We also show the corresponding cubic B-spline approximations (Fig.~\ref{fig:foba_examples}, bottom graph), which are good initial approximations of the original signal. In the next section, we refine these curves by applying a few steps of the B-spline VP algorithm.      

\subsection{B-spline VP algorithm}
\label{sec:bsplinevp}


Here, we consider the discrete VP problem, where the measured data is given in the form of vectors $\m{x},\m{f}\in\IR^N$ (i.e.,\ the $i$th coordinate $f_i$ represents the observed value at time $x_i$). The corresponding Hilbert space $(\mathcal{H},\left\langle \cdot,\cdot\right\rangle)$ is therefore identical to the vector space $\IR^N$ with the usual dot product. In this case, the projectors $P_{\mathcal{S(\boldsymbol{\alpha})}}\text{ and }P^{\bot}_{\mathcal{S(\boldsymbol{\alpha})}}$in Eq.~\eqref{eq:varpro_min} are equal to the matrices $\m{P}_{\boldsymbol{\Phi}(\boldsymbol{\alpha})}=\boldsymbol{\Phi}(\boldsymbol{\alpha})\boldsymbol{\Phi}^+(\boldsymbol{\alpha})$ and $\m{P}^{\bot}_{\boldsymbol{\Phi}(\boldsymbol{\alpha})}=\m{I}-\m{P}_{\boldsymbol{\Phi}(\boldsymbol{\alpha})}$, where $\boldsymbol{\Phi}_{ik}(\boldsymbol{\alpha})=\varphi_k(\boldsymbol{\alpha};x_i)\;(1\leq i \leq N,\; 1\leq k\leq n)$ denotes the matrix consisting of the uniformly sampled base functions, and $\boldsymbol{\Phi}^+(\boldsymbol{\alpha})$ is the Moore--Penrose pseudoinverse of $\boldsymbol{\Phi}(\boldsymbol{\alpha})$. Thus, the gradient $\nabla r_2$ is based on the Fr\'echet derivative of the matrix function $\boldsymbol{\Phi}:=\boldsymbol{\Phi}(\boldsymbol{\alpha})\in\IR^{N\times n}$ with respect to the vector variable $\boldsymbol{\alpha}\in\IR^m$. Golub and Pereyra \cite{golub-pereyra} showed that this can be interpreted as a three-dimensional tensor formed by the following matrix slabs:
\begin{equation*}
			\IR^{N\times n}\ni \m{D}_j:=\m{D}_j(\boldsymbol{\alpha})=\partial \boldsymbol{\Phi}(\boldsymbol{\alpha}) / \partial \alpha_j \quad (j=1,\ldots,m)\,\textrm{.}
	\label{eq:varpro_parder}
	\end{equation*}
If $\boldsymbol{\alpha}$ represents the vector of free knots, the columns of $\boldsymbol{\Phi}$ are the uniformly sampled B-splines $B_{\ell,k}(\tau_n(\boldsymbol{\alpha}); \cdot)$ of degree $\ell$, and the partial derivatives in $\m{D}_j$ are calculated according to Eq.~\eqref{eq:derBspline}. These matrices are sparse because the functions $B_{\ell,k}$ and the corresponding partial derivatives are zero outside their support $[t_k,t_{k+\ell+1}]$. For the sake of simplicity, we omit the vector of free parameters $\boldsymbol{\alpha}$ from the notations of these matrices. The $j$th coordinate of the gradient is then given by 
	\begin{equation}
		\frac{1}{2}\nabla r_{2_j}=\left[\m{J}_{:j}\right]^T \m{P}^{\bot}_\Phi \m{f}=\left[-\left(\m{P}^{\bot}_\Phi\m{D}_j\boldsymbol{\Phi}^{+} + \left(\m{P}^{\bot}_\Phi\m{D}_j\boldsymbol{\Phi}^{+}\right)^T \right)\m{f}\right]^T \m{P}^{\bot}_\Phi \m{f}\,\textrm{,}		
	\label{eq:varpro_grad}
	\end{equation}
where $\m{J}_{:j}$ denotes the $j$th column of the Jacobian. The two terms of the Jacobian matrix can be further simplified:
\begin{align}
\label{eq:kaufman_term1}
	\m{Kf}&=\m{P}^{\bot}_\Phi\m{D}_j\boldsymbol{\Phi}^{+}\m{f}=\m{P}^{\bot}_\Phi\m{D}_j\m{c}=\m{D}_j\m{c}-\m{P}_\Phi \m{D}_j\m{c}\,\textrm{,}\\
	\m{Lf}&=\left(\m{P}^{\bot}_\Phi\m{D}_j\boldsymbol{\Phi}^{+}\right)^T\m{f}=\left(\boldsymbol{\Phi}^{+}\right)^T \m{D}^T_j \m{P}^{\bot}_\Phi \m{f}=\left(\boldsymbol{\Phi}^{+}\right)^T \m{D}^T_j \m{r}\,\textrm{.}
\label{eq:kaufman_term2}
\end{align}
If the singular value decomposition (SVD) $\boldsymbol{\Phi}=\m{U\Sigma V}^T$ is given, then $\boldsymbol{\Phi}^+=\m{V}\boldsymbol{\Sigma}^+\m{U}^T$ and $\m{P}_{\Phi}=\m{U}\m{U}^T$. Kaufman \cite{kaufman} showed that the second term $\m{Lf}$ can be ignored because the residual $\m{r}$ becomes negligibly small near the solution. Note that only the first $\rank\left(\boldsymbol{\Phi}\right)$ number of columns rather than the full matrix $\m{U}$ must be computed. Hence, CPU time can be reduced when $n \ll N$ (e.g.,\ compressing a signal consisting of $N$ samples by storing only the $n$ coefficients of its least-squares approximation). For this reason we use the economy-sized SVD decomposition in combination with sparse-matrix computations. Note that our implementation is based on the work of O'Leary and Rust \cite{varpro_matlab}. We adapted their algorithm to B-splines with free knots by utilizing the special properties mentioned above. Section~\ref{sec:compECG} presents experimental results which show that the proposed algorithm halves execution time compared to the former VP implementation.


\section{Numerical experiments}
\label{sec:experiments}

\subsection{Approximating synthetic data}
The evolution of B-spline free-knot optimization methods dates back to the 70s, when researchers used various test functions to demonstrate the efficiency of these algorithms. For instance, one of the most popular is the titanium heat dataset, which measures the properties of titanium as functions of temperature. This set of test functions has been extended by numerous authors over the past few decades. However, the performance evaluations of previous algorithms are not comparable (e.g.,\ the authors used different formulas to quantify the numerical errors of the approximations). Another issue concerns computational complexity, which can be quantified in many ways, for instance, by measuring execution time or the number of FLOPS, iterations or function evaluations. In most cases no implementational details of these algorithms were published, which makes it difficult to provide a fair comparison. To overcome these problems, we considered methods for which the proper number of function evaluations or the exact error formula was provided by the authors. 

In this experiment, we first estimated the initial knots by Alg.~\ref{alg:knotpred} using $\ell_p\;(p=1,2,\infty)$ norm solutions. We then applied a few iterations of the B-spline VP method. From the three initialization strategies, we chose that for which the VP optimization achieved the smallest approximation error (see Fig.~\ref{fig:examples}). The list of test functions and error measures can be found in Tab~\ref{tab:synt_fun}. In some cases, we used a noise signal $w$ superimposed on the original data, which was simulated by uniformly distributed random numbers within a specific interval. The performance of the proposed method was compared to various optimization strategies, such as the Levenberg--Marquardt method \cite{lindstrom}, the Lasso algorithm \cite{lasso}, global search techniques \cite{uyarulker}, and genetic algorithms (GA) \cite{yoshimoto}. Tab.~\ref{tab:synt_test} summarizes the results, where \textsl{Nit} and \textsl{Nfe} denote the number of iterations and the number of function evaluations during optimization, respectively, and $n+1$ is the number of knots (i.e.\ $n-1$ free knots plus $2$ boundary knots). 

We conclude that our approach provides a good alternative to the other algorithms mentioned in this study. Our method either outperformed competing methods or required fewer iterations to find a stationary point close to the optimal fitness value. Although we applied a gradient-based local search method initialized by Alg.~\ref{alg:knotpred}, the corresponding estimations of the knots can also be used in global search techniques. For instance, the initial population of GAs can be formed by individuals which are based on knots predicted in different $\ell_p$ norms $(p=1,2,\infty)$. The experiment using synthetic data also shows that our algorithm is able to deal with coalescent knots, discontinuous functions, cusps and noise. 

\begin{figure}[htb!]
\centering
  \subfigure[$f_1(x)+w(x)$]{
  \includegraphics[scale=0.335, trim=140 265 154 270, clip]{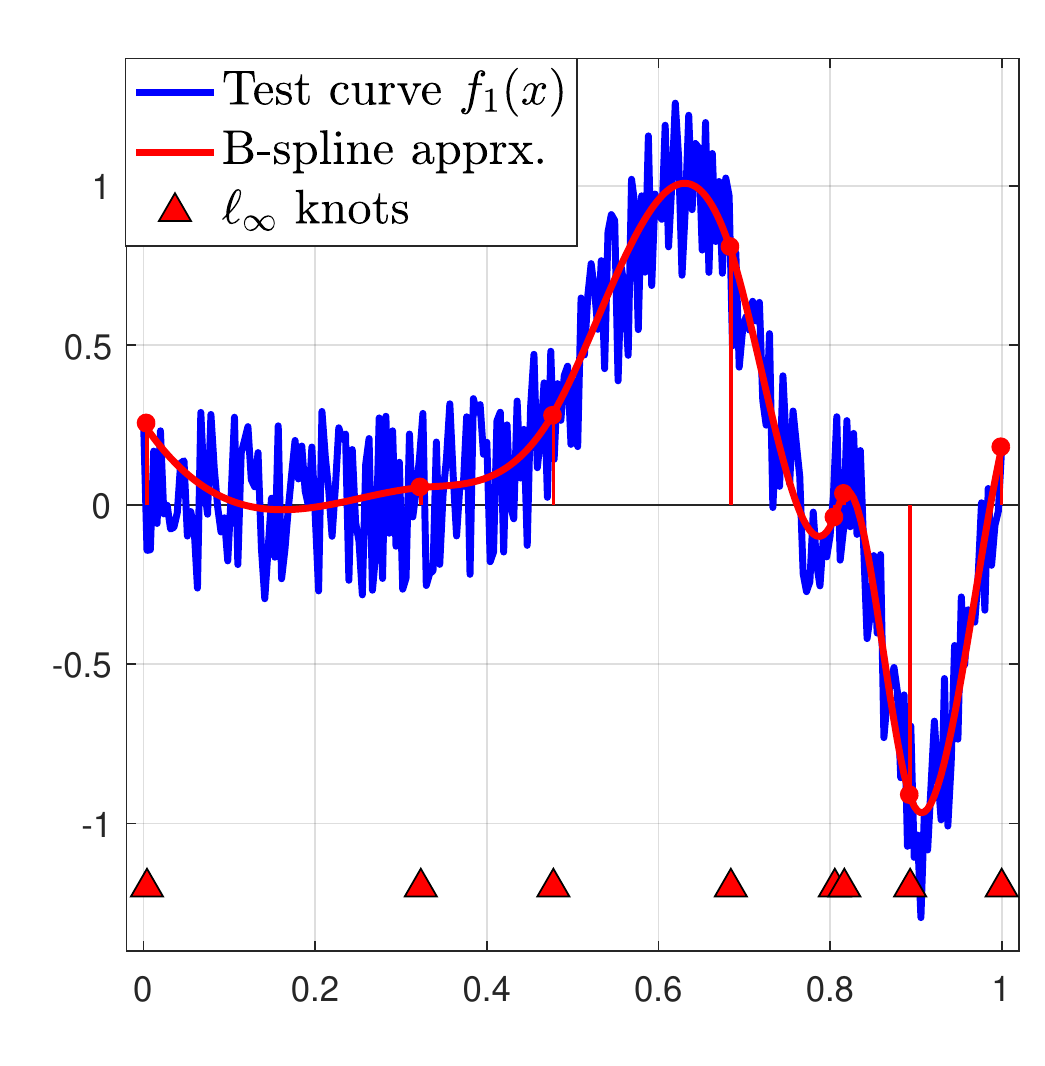}
  \label{fig:f1}
  }\hspace{2mm}
  \subfigure[Titanium heat data $f_2(x)$]{
  \includegraphics[scale=0.335, trim=140 265 154 270, clip]{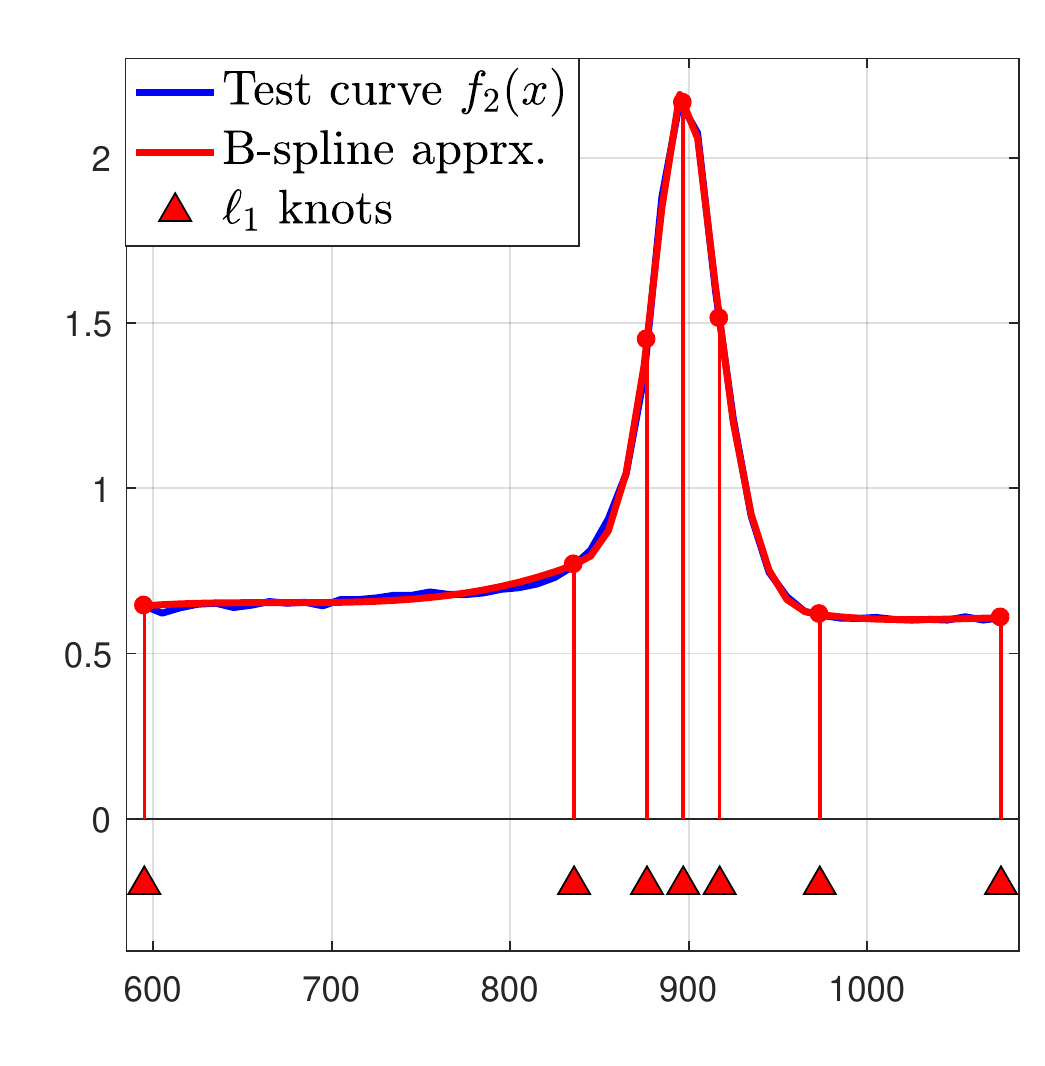}
  \label{fig:f2}
  }\hspace{2mm}
  \subfigure[$f_3(x)$]{
  \includegraphics[scale=0.335, trim=140 265 154 270, clip]{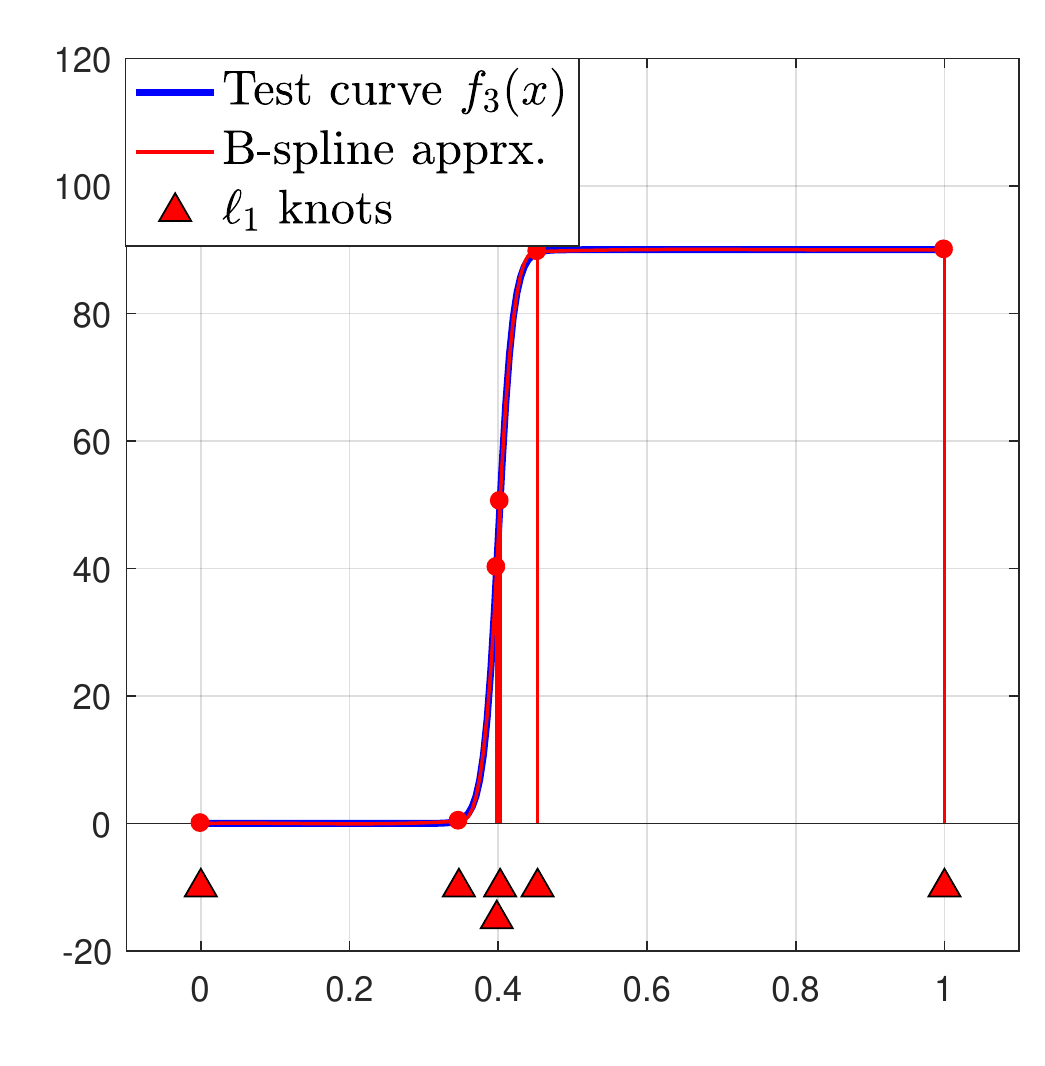}
  \label{fig:f3}
  } 	
	
  \subfigure[$f_4(x) + w(x)$]{
  \includegraphics[scale=0.335, trim=140 265 154 270, clip]{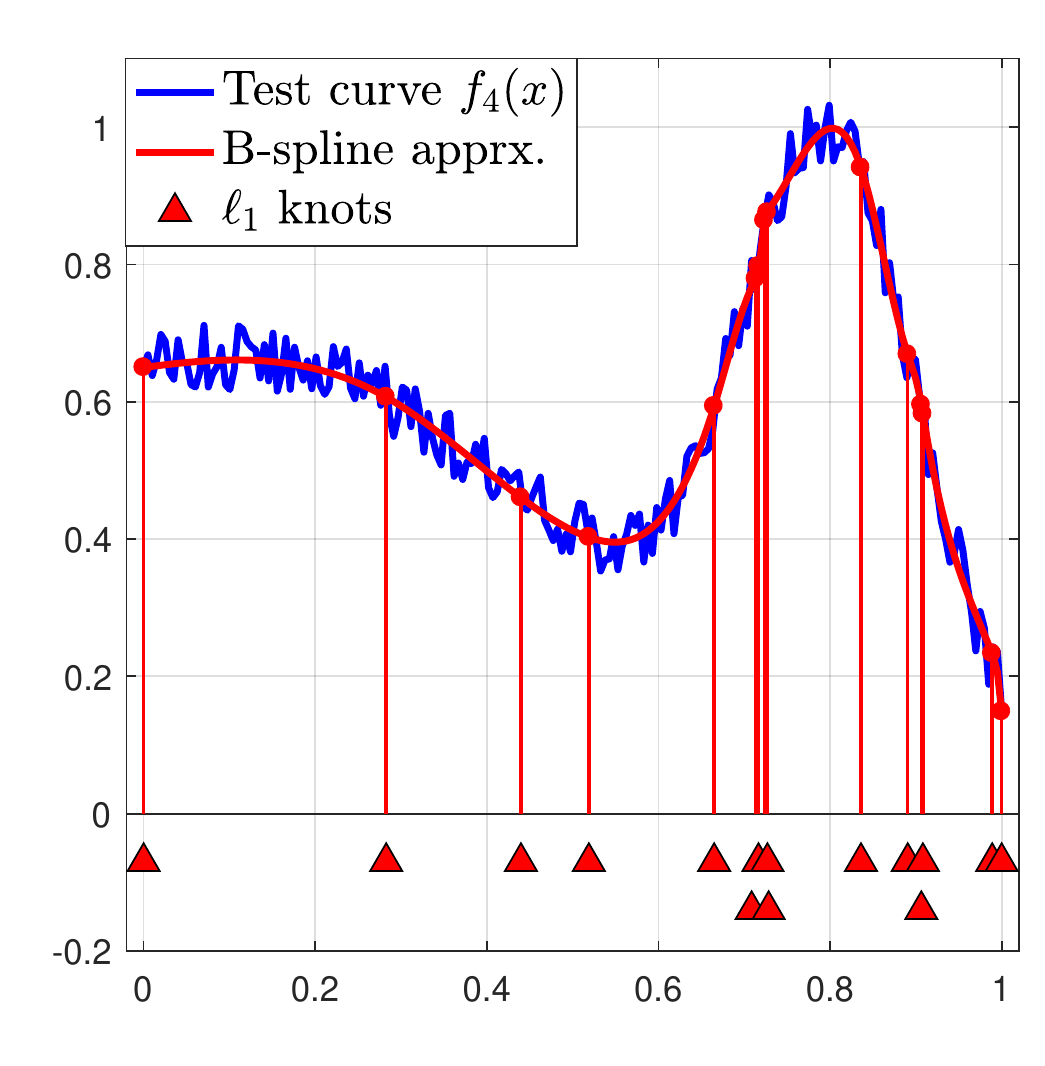}
  \label{fig:f4}
  }\hspace{2mm}
  \subfigure[$f_5(x)$]{
  \includegraphics[scale=0.335, trim=140 265 154 270, clip]{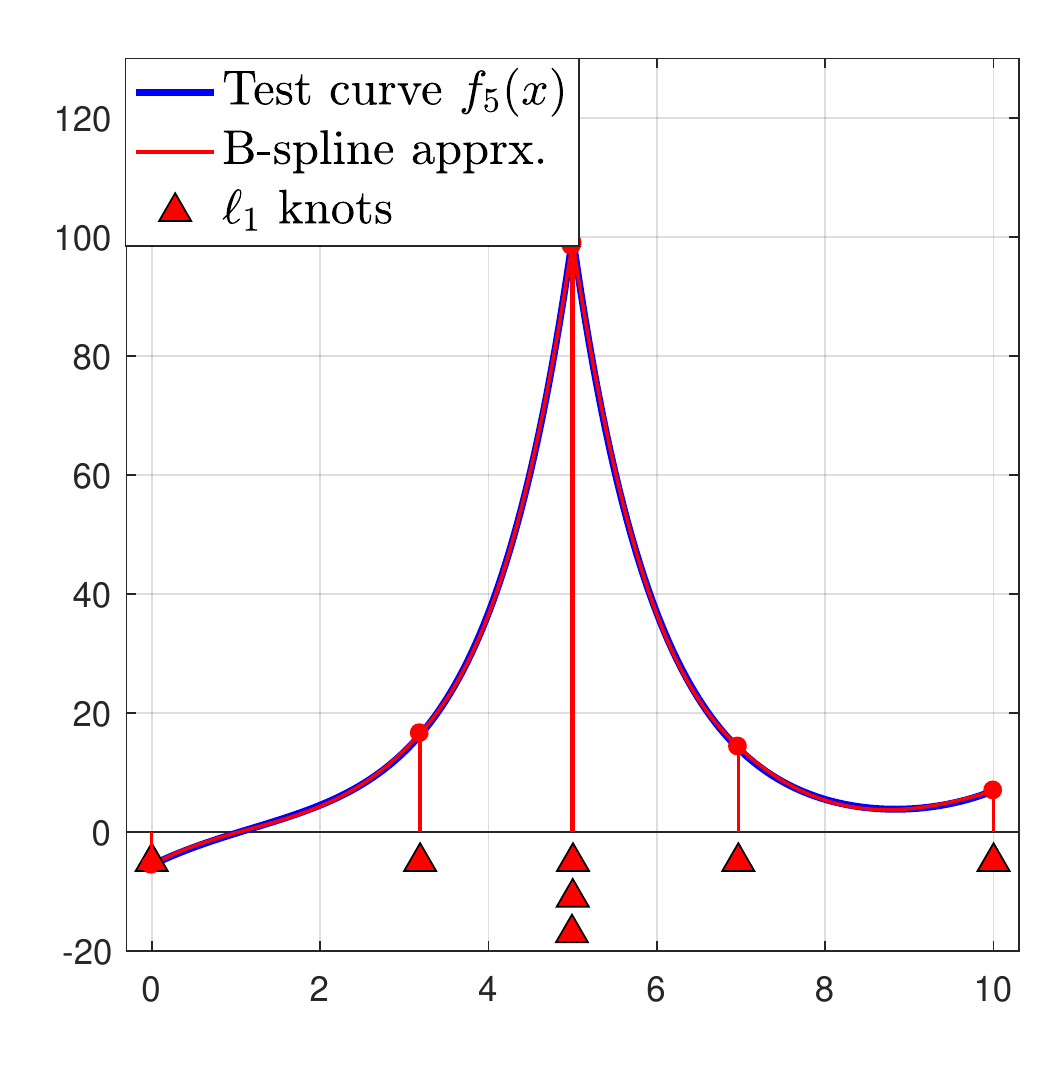}
  \label{fig:f5}
  }\hspace{2mm}
  \subfigure[$f_6(x)$]{
  \includegraphics[scale=0.335, trim=140 265 154 270, clip]{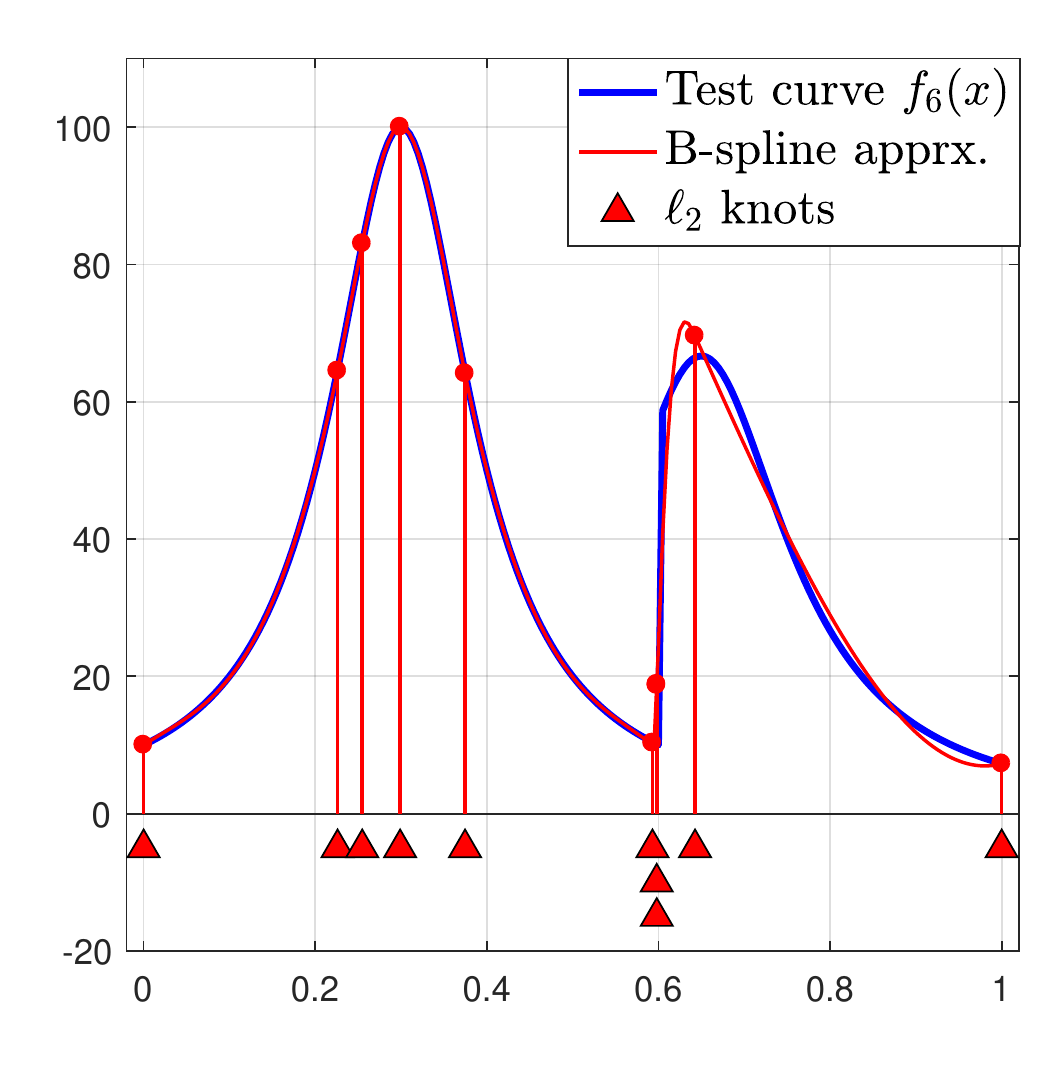}
  \label{fig:f6}
  }	
\caption{Test functions and their cubic B-spline approximations.}
\label{fig:examples}
\end{figure}

\begin{table}[htb!]
\begin{center}
\caption{Definitions of the synthetic data and the error measures.}
\scalebox{0.70} { { \renewcommand{\arraystretch}{1.2} 
  \begin{tabular}{|l l|}
		\hline
			\multicolumn{2}{|c|}{\textbf{Test functions}}\\\hline\hline
	$f_1(x)=(\sin 2\pi x^3)^3$ & $(x\in[0,1])$\\\hline
	$f_2(x)=$  Titanium heat data & \\\hline
	$f_3(x)=90/(1+e^{-100\cdot(x-0.4)})$ & $(x\in[0,1])$\\\hline
	$f_4(x)=\frac{1}{2.3935} \left(1.5\cdot e^{-\frac{(t-0.1)^2}{0.3}}+0.1 \cdot e^{-\frac{(t-0.5)^2}{2}}+2\cdot e^{-\frac{(t-0.8)^2}{0.02}}\right)$ & $(x\in[0,1])$\\\hline
	$f_5(x)=100/e^{\left|x-5\right|}+(x-5)^5/500$ & $(x\in[0,10])$\\\hline
	$f_6(x)= 
	    \begin{cases}
				1/(0.01+(x-0.3)^2)\text{,} & x< 0.6 \\
				1/(0.015+(x-0.65)^2)\text{,} & x\geq 0.6 
			\end{cases}$ 
			& $(x\in[0,1])$\\\hline\hline	
	\multicolumn{2}{|c|}{\textbf{Error measures}}\\\hline\hline			
	\multicolumn{2}{|l|}{$\text{BRE}:=\left(\frac{1}{N-1} \sum_{1}^{N} v_i (\m{f}_i-\m{\widetilde{f}}_i)^2\right)^{1/2} \qquad (v_1=v_N=\frac{1}{2}$ and $v_i=1$ otherwise)}\bigstrut\\\hline	
	\multicolumn{2}{|l|}{$\text{MSE}:=\frac{1}{N}\cdot \text{RSS}:=\frac{1}{N}\cdot \|\m{f}-\m{\widetilde{f}}\|_2^2$}\bigstrut\\\hline
	\multicolumn{2}{|l|}{$\text{BIC}~:=N\cdot\ln{\|\m{f}-\m{\widetilde{f}}\|_2^2} + \ln{\big(N\cdot(2(n-1)+\ell+1)\big)}$}\bigstrut\\\hline
  \end{tabular}
}}
\label{tab:synt_fun}
\end{center}
\end{table}

\begin{table}[!htb]
\begin{center}
\caption{Evaluation of the performance on synthetic data.}
\scalebox{0.6} { { \renewcommand{\arraystretch}{1.2} 
  \begin{tabular}{| l l r |c | c | c || c | c | c || c | c | c | c |}
		\hline
			\multicolumn{6}{|c||}{\textbf{Data}} & \multicolumn{3}{|c||}{\textbf{Other methods}} & \multicolumn{4}{|c|}{\textbf{Proposed method}} \\ \hline
			\multicolumn{3}{|c||}{\textbf{Signal}} & \textbf{Noise} ($\m{w}$) & \textbf{N} & \textbf{n+1} & \textbf{Nit} & \textbf{Nfe} & \textbf{Error} & \textbf{Nit} & \textbf{Nfe} & \textbf{Error} & \textbf{Mod} \\ 
		\hline\hline
		\multicolumn{13}{|l|}{Tests measuring residual sum of squares (RSS)} \\ \hline
			$f_1(x)\;\;$ & Fig.~\ref{fig:f1} &\cite{lindstrom} & $[-0.3,0.3]$ & 256 & 8 & 26 & 63 & 19.5 & 4 & 5 & \textbf{7.9950} & $\ell_\infty$ \\ \hline
			$f_2(x)\;\;$ & &\cite{lindstrom} & -- & 49 & 9 & 156 & 347 & \textbf{0.00138} & 20 & 21 &0.00209& $\ell_\infty$ \\ \hline\hline
			\multicolumn{13}{|l|}{Tests measuring (MSE)}\\ \hline
			$f_3(x)\;\;$ & &\cite{uyarulker} &  --   & 101 & 15 & 40 & $40<$ & 0.5140 & 4 & 5 & \textbf{0.00019} & $\ell_2$ \\ \hline
			$f_4(x)\;\;$ & Fig.~\ref{fig:f4} &\cite{uyarulker} & $[-0.05,0.05]$  & 200 & 15 & 5 & $5<$ & 0.00718 & 4 &5 &\textbf{0.00082}& $\ell_1$ \\ \hline \hline
			\multicolumn{13}{|l|}{Tests measuring de Boor and Rice error (BRE)}\\ \hline\hline
			$f_2(x)\;\;$ & Fig.~\ref{fig:f2} &\cite{jupp}& -- & 49 & 7 & 8 & $9$ & \textbf{0.01227} & 5 &6 & 0.01325 & $\ell_1$ \\ \hline
			$f_2(x)\;\;$ & &\cite{uyarulker} & -- & 49 & 7 & 5 & $5<$ & \textbf{0.00942} & 5 &6 & 0.01325 & $\ell_1$ \\ \hline
			$f_2(x)\;\;$ & &\cite{yuan} & -- & 49 & 8 & unk. & unk. & 0.01174 & 6 &7 & \textbf{0.00874} & $\ell_\infty$ \\ \hline
			\multicolumn{13}{|l|}{Tests measuring Bayes Information Criterion (BIC)}\\ \hline
			$f_3(x)\;\;$ & Fig.~\ref{fig:f3} &\cite{yoshimoto}  & --  & 201 & 6 & 30 & 1500 & 1189 & 7 & 8 &\textbf{332}& $\ell_1$ \\ \hline
			$f_5(x)\;\;$ & Fig.~\ref{fig:f5} &\cite{yoshimoto} & -- & 201 & 7 & 30 & 1500 & 1188 &14 & 15  & \textbf{471} & $\ell_\infty$ \\ \hline			
			$f_6(x)\;\;$ & Fig.~\ref{fig:f6} &\cite{yoshimoto} & --  & 201 & 10 & 129 & 6450 & \textbf{1181} & 19 & 20 &1491& $\ell_2$ \\ \hline
  \end{tabular}
}}
\label{tab:synt_test}
\end{center}
\end{table}

\subsection{Compressing real-world ECG data}
\label{sec:compECG}

We demonstrate the efficiency of our method in the context of ECG compression, where a vast amount of data is to be processed within reasonable time. In this case, the original signal is represented by $N$ samples, but only a small portion of the data (i.e.\ the coefficients and the knots) is stored. ECGs are quasi-periodic signals in which each period captures the contraction and relaxation pattern of the heart muscles. The recordings include several channels which measure the potential difference between the electrodes placed on the human body. The PhysioNet MIT-BIH Arrhythmia Database \cite{PhysioNet} is a standard dataset in this field, prepared for signal processing purposes. It contains 48 half-hour-long two-channel recordings sampled at $360$ Hz. We compressed an 11-hour portion of the dataset called DS1, as recommended by de Chazal et al.~\cite{ds1} for designing and testing heartbeat classification algorithms. 
In DS1, there are 22 recordings -- including both normal ($45868$) and abnormal ($5152$) heartbeats -- from various patients. This allowed us to examine also how inter- and intra-patient variability affects knot-prediction, approximation error, optimization and compression rate.

The performance of compression algorithms was evaluated by means of the reconstruction error and the compression ratio (CR). The former measures the numerical error of the approximation, and the latter quantifies the reduction in size of the original data. For a discrete time signal $\m{f}\in\IR^N$ with $N$ samples, the reconstruction error and CR can be defined as follows:
\begin{equation}
	\epsilon_p=\frac{\|\m{f}-\m{\widetilde{f}}\|_p}{\|\m{f}-\m{\overline{f}}\|_p}\times 100\,\textrm{,}\qquad \CR= \frac{\text{Size of the uncompressed data}}{\text{Size of the compressed data}}=\frac{N}{M}\,\textrm{,}
\label{eq:prd_cr_def}
\end{equation}
where $\m{\widetilde{f}}$ denotes the approximation and $\m{\overline{f}}$ is the mean of the original signal. The quantity $\epsilon_p$ is a kind of normalized relative error of the approximation. For $p=2$, the normalized percent root-mean-square difference (PRDN) can be calculated, which is a standardized measure of the reconstruction error in ECG signal processing. In order to measure the CR, we need to know the size of the compressed data. For B-splines of degree $\ell$, it is equal to $M=2n+\ell+1$ since not only the positions of the interior knots $\boldsymbol{\alpha}\in\IR^{n-1}$, but also the two boundary knots and the coefficient vector $\m{c}\in\IR^{n+\ell}$ must be stored. 

Since the ECG recordings in \cite{PhysioNet} were annotated manually by two or more cardiologists, we could segment the signals into heartbeats. The optimal vector of free knots for each heartbeat was estimated in the following way:
\begin{enumerate}
	\item Predict the knots by Alg.~\ref{alg:knotpred};
	\item refine the prediction by the B-spline VP method;
	\item evaluate PRDN and CR in the final iteration of the VP method.
\end{enumerate} 

Before we could process the ECG signals, we had to determine two parameters of the compression algorithm: the number of free knots and the number of iterations in the VP optimization. To this end, we took the first two minutes of the recordings in DS1 (this allowed us to portion $22\times 2$ minutes of the data). Using Alg.~\ref{alg:knotpred}, we then predicted the positions of $50$ knots for each heartbeat. Fig.~\ref{fig:numofknots_foba} shows the average of the error $\epsilon_p\; (p=1,2,\infty)$ for each additional knot. It can be seen that using more than $25$ knots does not reduce the error significantly. Note that, although this is just a first-order B-spline approximation, it gives a good estimate of the number of free knots.
Further, the FOBA calculations are very simple due to the explicit form of the coefficients (see Eqs.~\eqref{eq:l2coeffs}-\eqref{eq:recall_lpcoeff}). The average execution times of the FOBA methods were only $55,\, 3$ and $1.5$ minutes for $p=1,2,\infty$, respectively. In particular, the $\ell_2$ and $\ell_\infty$ variants of Alg.~\ref{alg:knotpred} were sufficient to estimate the number of free knots within reasonable time.

\begin{figure}[!htb]
\centering
  \subfigure[B-spline approximations of an ECG.]{
  \includegraphics[scale=0.49, trim=130 265 120 265, clip]{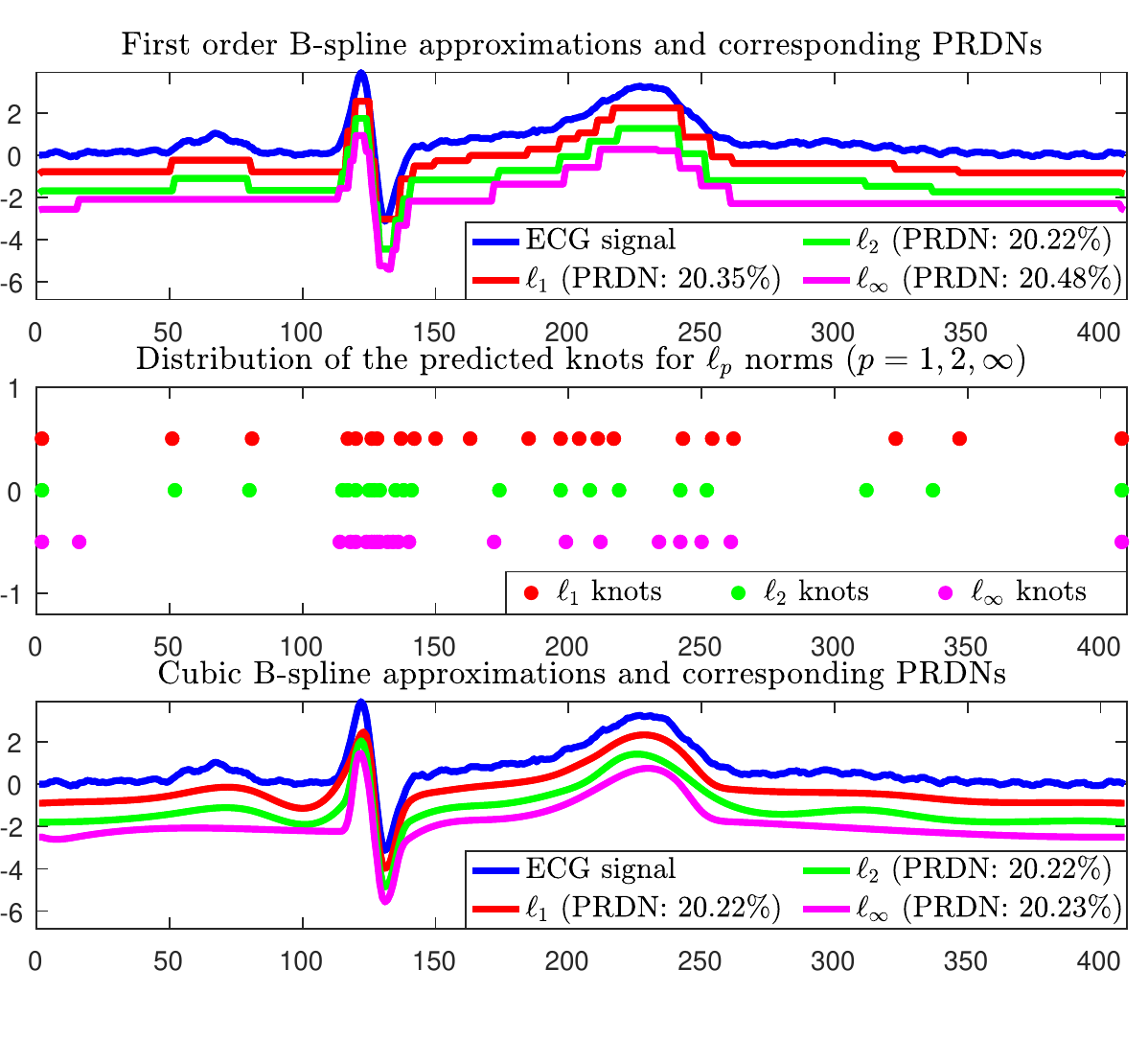}
  \label{fig:foba_examples}
  }\hspace{2mm}
	\subfigure[Estimation of the number of free knots.]{
  \includegraphics[scale=0.50, trim=145 265 150 270, clip]{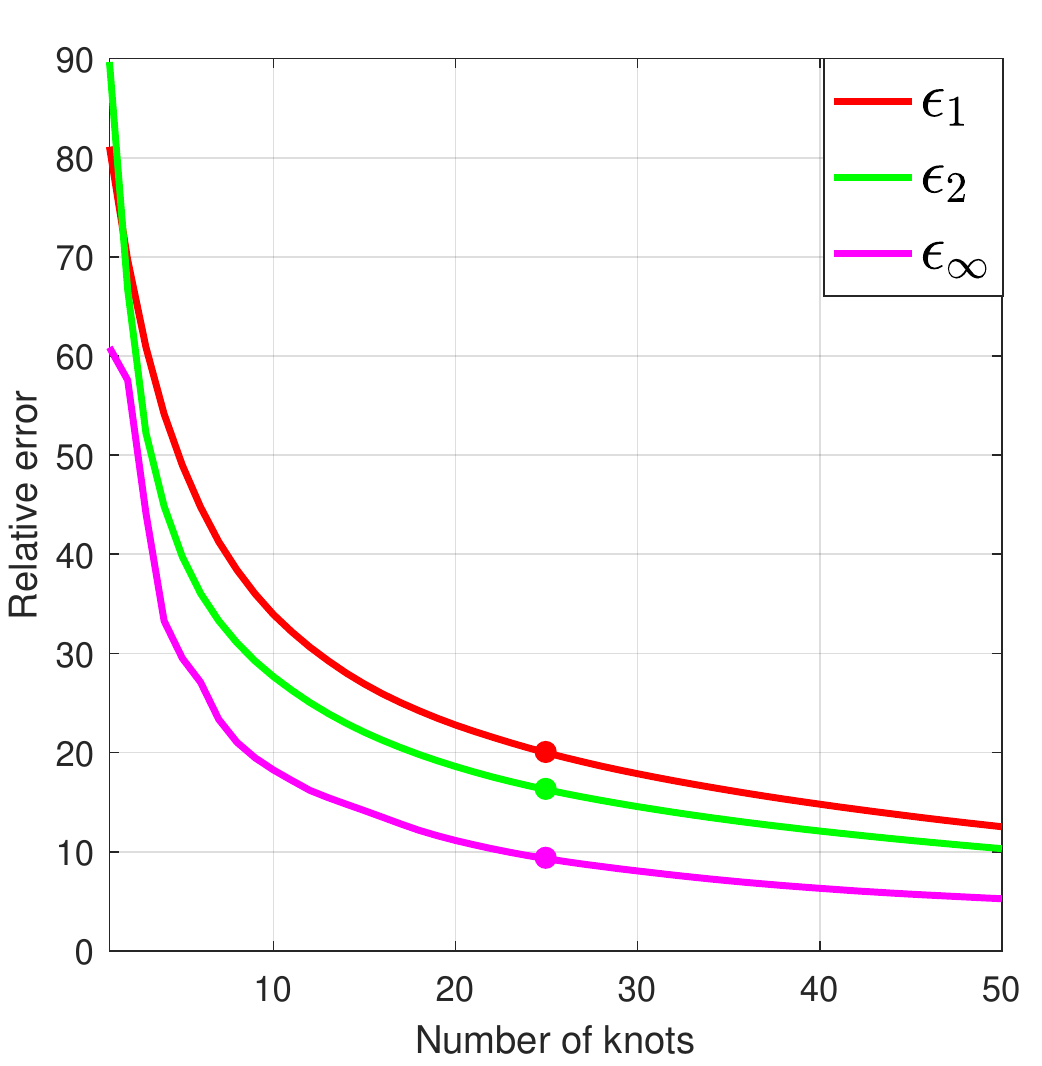}
  \label{fig:numofknots_foba}
  } 
\caption{Properties of the knot-prediction algorithm.}
\label{fig:examples_exectimes}
\end{figure}

We also analyzed the speed of convergence of the B-spline VP method on dataset DS1. Let us suppose that the sequence of vectors consisting of free knots $\boldsymbol{\alpha}^{(k)}$ converges to $\boldsymbol{\alpha}^{*}$; the usual definition of the asymptotic rate $\mu$ and the order $\rho$ of convergence then have the form:
\begin{equation*}
	\lim_{k\rightarrow\infty} \frac{\left\|\boldsymbol{\alpha}^{(k+1)}-\boldsymbol{\alpha}^*\right\|}{\left\|\boldsymbol{\alpha}^{(k)}-\boldsymbol{\alpha}^{*}\right\|^{\rho}} \approx \lim_{k\rightarrow\infty} \frac{\left\|\boldsymbol{\alpha}^{(k+1)}-\boldsymbol{\alpha}^{(k)}\right\|}{\left\|\boldsymbol{\alpha}^{(k)}-\boldsymbol{\alpha}^{(k-1)}\right\|^{\rho}}=\lim_{k\rightarrow\infty} \frac{\left\|\boldsymbol{\varepsilon}^{(k+1)}\right\|}{\left\|\boldsymbol{\varepsilon}^{(k)}\right\|^{\rho}}=\mu\,\textrm{.}
\label{eq:rate_order}
\end{equation*}
Assuming that we are close to the limit, the logarithm of the right-hand side can be written as
\begin{equation*}
	\log\|\boldsymbol{\varepsilon}^{(k+1)}\|\approx \rho\cdot \log\|\boldsymbol{\varepsilon}^{(k)}\| + \log \mu\,\textrm{.}
\end{equation*}
Considering this form, we estimated the parameters $\rho\text{ and }\mu$ by applying linear regression to the data $(\log\|\boldsymbol{\varepsilon}^{(k)}\|,\log\|\boldsymbol{\varepsilon}^{(k+1)}\|)$ for $k=1,\ldots,100$. The slope and the $y$-intercept of the fitted polynomial give $\rho$ and $\log \mu$, respectively. According to the overall average of the results in Tab.~\ref{tab:convtest_avg}, the order $\rho$ is close to linear, with the rate of convergence $\mu \approx 0.75$. It also turned out that ignoring the second term $\m{L}$ in the Jacobian (see Eqs.~\eqref{eq:kaufman_term1}-\eqref{eq:kaufman_term2}) did not change the final results significantly. In fact, the differences in overall average PRDN and required number of iterations between the two variants of the B-spline VP method in Tab.~\ref{tab:convtest_avg} are negligible. For this reason, we computed the full Jacobian matrix in our further experiments.

\begin{table}[!tb]
\centering
	\caption{Average results of the convergence test for $100$ B-spline VP iterations.}
	\scalebox{0.68}{
\begin{tabular}{|l||cc||cccc||cccc|}
\multicolumn{1}{c}{\textbf{}} & \multicolumn{2}{c}{\textbf{Initial Jacobian}} & \multicolumn{4}{c}{\textbf{VP with full Jacobian}} & \multicolumn{4}{c}{\textbf{VP with partial Jacobian}} \bigstrut\\\hline
\textbf{Mod} & $\left\|\m{K}\right\|_2/\left\|\m{J}\right\|_2$ & $\left\|\m{L}\right\|_2/\left\|\m{J}\right\|_2$ & $\boldsymbol{\mu}$ & $\boldsymbol{\rho}$ & \textbf{Nit} & \textbf{PRDN} ($\%$) & $\boldsymbol{\mu}$ & $\boldsymbol{\rho}$ & \textbf{Nit} & \textbf{PRDN} ($\%$) \bigstrut[t]\\\hline
$\ell_1$ & 0.85 & 0.57 & 0.77 & 0.87 & 3.24 & 5.95 & 0.64 & 0.75 & 2.97 & 6.11 \\
$\ell_2$  & 0.86 & 0.56 & 0.75 & 0.88 & 2.98 & 5.99 & 0.62 & 0.76 & 2.72 & 6.28 \\
$\ell_\infty$ & 0.86 & 0.55 & 0.74 & 0.89 & 2.87 & 6.49 & 0.60 & 0.77 & 2.58 & 6.89\\
\hline
\end{tabular}
}
\label{tab:convtest_avg}
\end{table}

Next, we had to find the proper number of iterations for the B-spline VP optimization. To this end, we used a residual-based termination criterion:
\begin{equation*}
\left\|\left(\m{P}_{\boldsymbol{\Phi}(\boldsymbol{\alpha}^{(k+1)})}-\m{P}_{\boldsymbol{\Phi}(\boldsymbol{\alpha}^{(k)})}\right)\m{f}\right\|_2<10^{-1}\,\textrm{,}
\end{equation*}
where $\boldsymbol{\alpha}^{(k)}\in\IR^{n-1}$ denotes the vector of free knots at the $k$th iteration, and $\m{P}_{\boldsymbol{\Phi}(\boldsymbol{\alpha}^{(k)})}$ is the corresponding cubic B-spline projector matrix. In Tab.~\ref{tab:convtest_avg}, the column \textsl{Nit} lists the numbers of iterations required to satisfy this condition. Based on these results, we stopped the cubic B-spline VP optimization after $4$ iterations. 

Using $n+1=25$ knots and $4$ VP iterations, we repeated the tests on the whole DS1 dataset and compared the performance with that of other algorithms. In all experiments, we measured the execution time on a computer equipped with Intel(R) Core(TM) i7-6700 @ 3.40GHz CPU. Before we analyze the results, let us list the methods we compared and define the corresponding abbreviations in Tab.~\ref{tab:aprx}: 
\begin{itemize}
	\item \textbf{CR}: According to Eq.~\eqref{eq:prd_cr_def}, the compression ratio (CR) for the whole recording is $N/M$, where $M=\text{number of heartbeats} \times (2\cdot 24 + 3 + 1)$ and $N$ denotes the overall number of samples.
	\item \textbf{FOBA PRDN}: We predicted the knots by FOBA, and we computed the PRDN of the corresponding cubic B-spline approximation.
	\item \textbf{VP PRDN}: The estimated knots were refined by applying $4$ B-spline VP iterations. In the final step, we evaluated the PRDN of the cubic B-spline approximation, which represents the reconstructed signal.
	\item \textbf{KR}: This is a \textsl{Knot-Reduction} procedure \cite{ekg} in which the knot vector is initialized with all the sample points. In each iteration, the number of knots is reduced by removing the knot whose absence increases the mean squared error the least. 
	\item \textbf{UVP}: This denotes the same B-spline VP algorithm as used in our approach, but initialized with uniformly distributed knots. 
	\item \textbf{RVP}: This is a B-spline VP algorithm with $4$ iterations initialized with randomly distributed knots.
\end{itemize}

\begin{table}[!htb]
\centering
	\caption{Performance of the algorithms tested. The best PRDNs of each group of columns are shown in bold face, and the lowest PRDN of each row is underlined.
}
	\scalebox{0.58}{
\begin{tabular}{|c|c||rrr||rrr|rrr|ccccc|}
\multicolumn{1}{c}{\textbf{}} & \multicolumn{1}{c}{\textbf{}} & \multicolumn{3}{c}{\textbf{FOBA PRDN ($\%$)}} & \multicolumn{3}{c}{\textbf{VP PRDN ($\%$)}} & \multicolumn{3}{c}{\textbf{Other PRDN ($\%$)}} & \multicolumn{5}{c}{\textbf{Execution time (min)}} \bigstrut\\\hline
\textbf{Rec.} & \textbf{CR} & $\boldsymbol{\ell_1}$ & $\boldsymbol{\ell_2}$ & $\boldsymbol{\ell_\infty}$ & $\boldsymbol{\ell_1}$ & $\boldsymbol{\ell_2}$ & $\boldsymbol{\ell_\infty}$ & \textbf{KR} & \textbf{UVP} & \textbf{RVP} & $\boldsymbol{\ell_1}$ & $\boldsymbol{\ell_2}$ & $\boldsymbol{\ell_\infty}$ & \textbf{VP} & \textbf{KR}  \bigstrut[t]\\\hline
101 & 6.72 & 11.89 & 10.12 & \textbf{10.03} & 6.62 & \textbf{6.57} & 6.85 & \underline{5.38} & 20.44 & 30.43 & 2.59 & 0.13 & 0.05 & 4.11 & 67.62 \\
106 & 6.18 & 9.78 & 8.71 & \textbf{8.64} & 6.43 & \textbf{6.08} & 6.23 & \underline{4.91} & 13.59 & 20.84 & 2.67 & 0.14 & 0.06 & 4.17 & 68.18 \\
108 & 7.11 & 19.63 & 18.71 & \textbf{17.44} & 15.17 & 14.42 & \textbf{14.06} & \underline{12.32} & 18.88 & 24.88 & 2.75 & 0.14 & 0.05 & 3.81 & 73.75 \\
109 & 4.96 & 8.37 & 8.03 & \textbf{7.59} & \textbf{4.29} & 5.38 & 5.23 & \underline{2.48} & 4.29 & 5.65 & 2.63 & 0.13 & 0.07 & 4.96 & 46.58 \\
112 & 4.94 & 9.65 & \textbf{9.64} & 9.65 & \textbf{7.00} & 7.06 & 7.47 & \underline{5.92} & 8.26 & 15.74 & 2.72 & 0.14 & 0.07 & 5.02 & 46.08 \\
114 & 6.67 & 17.11 & \textbf{15.70} & 15.98 & 14.10 & \textbf{13.83} & 13.91 & \underline{12.35} & 40.60 & 40.38 & 2.59 & 0.13 & 0.06 & 3.96 & 68.90 \\
115 & 6.42 & 10.79 & \textbf{8.61} & 8.81 & \textbf{4.91} & 5.30 & 6.37 & \underline{3.95} & 30.24 & 40.91 & 2.60 & 0.13 & 0.06 & 4.05 & 63.86 \\
116 & 5.20 & 8.25 & \textbf{8.14} & 10.00 & 4.99 & \textbf{4.93} & 6.27 & \underline{4.22} & 13.52 & 21.88 & 2.57 & 0.13 & 0.06 & 4.71 & 48.99 \\
118 & 5.51 & 12.13 & 11.55 & \textbf{10.82} & 8.63 & 8.21 & \textbf{8.16} & \underline{6.24} & 16.51 & 19.54 & 2.65 & 0.13 & 0.06 & 4.64 & 52.88 \\
119 & 6.31 & 14.50 & 9.33 & \textbf{7.57} & 7.97 & \textbf{4.77} & 5.00 & \underline{3.57} & 15.53 & 22.64 & 2.56 & 0.13 & 0.05 & 4.10 & 69.36 \\
122 & 5.07 & 7.53 & \textbf{7.37} & 9.22 & 5.34 & \textbf{5.19} & 5.47 & \underline{4.05} & 9.12 & 12.84 & 2.52 & 0.13 & 0.07 & 4.85 & 47.05 \\
124 & 7.74 & \textbf{8.29} & 9.23 & 11.81 & 4.92 & \textbf{4.67} & 7.18 & \underline{3.46} & 14.22 & 20.17 & 2.66 & 0.14 & 0.04 & 3.49 & 80.58 \\
201 & 6.38 & 8.28 & \textbf{8.09} & 8.65 & 5.79 & \textbf{5.75} & 6.33 & \underline{4.71} & 11.28 & 17.98 & 2.66 & 0.14 & 0.05 & 4.03 & 80.41 \\
203 & 4.21 & 9.90 & \textbf{9.74} & 9.84 & 7.72 & \textbf{7.62} & 7.73 & \underline{6.06} & 8.54 & 10.79 & 2.66 & 0.17 & 0.08 & 5.69 & 43.12 \\
205 & 4.72 & 9.80 & 9.68 & \textbf{9.42} & 6.56 & \textbf{6.44} & 6.64 & \underline{5.44} & 14.84 & 25.68 & 2.86 & 0.14 & 0.07 & 5.20 & 43.83 \\
207 & 6.50 & 8.08 & 7.71 & \textbf{7.28} & 5.77 & 5.58 & \textbf{5.57} & \underline{4.89} & 7.62 & 9.77 & 2.39 & 0.13 & 0.05 & 3.45 & 151.51 \\
208 & 4.25 & 7.62 & 7.06 & \textbf{7.03} & 5.15 & \textbf{4.95} & 5.18 & \underline{4.18} & 6.61 & 11.86 & 2.66 & 0.15 & 0.08 & 5.70 & 40.88 \\
209 & 4.18 & 14.06 & 11.69 & \textbf{10.70} & 8.86 & \textbf{8.17} & 8.39 & \underline{6.46} & 11.03 & 22.70 & 2.71 & 0.15 & 0.08 & 5.74 & 37.91 \\
215 & 3.74 & 12.61 & 12.30 & \textbf{11.59} & \textbf{8.10} & 8.24 & 8.45 & \underline{6.22} & 8.09 & 13.93 & 3.04 & 0.17 & 0.09 & 6.41 & 35.05 \\
220 & 6.13 & 10.71 & \textbf{8.37} & 8.78 & \textbf{5.47} & 5.47 & 6.30 & \underline{4.20} & 24.18 & 38.56 & 2.60 & 0.12 & 0.06 & 4.15 & 61.84 \\
223 & 4.82 & \textbf{5.54} & 5.78 & 5.63 & \textbf{3.52} & 3.62 & 3.89 & \underline{2.85} & 5.58 & 10.86 & 2.66 & 0.14 & 0.07 & 5.09 & 46.31 \\
230 & 5.56 & 8.56 & \textbf{8.50} & 10.57 & \textbf{4.88} & 5.33 & 7.63 & \underline{3.81} & 11.60 & 23.54 & 2.58 & 0.13 & 0.06 & 4.47 & 55.29 \\\hline
\textbf{Avg.} & 5.61 & 10.60 & \textbf{9.73} & 9.87 & 6.92 & \textbf{6.71} & 7.20 & \underline{5.35} & 14.30 & 20.98 & 2.65 & 0.14 & 0.06 & 4.63 & 60.45 \\
\hline
\end{tabular}
}
\label{tab:aprx}
\end{table}

In Tab.~\ref{tab:aprx}, it can be seen that the initial cubic B-spline approximation has an average PRDN of $10\%$, which was improved by more than $3\%$ although applying only $4$ iterations of the B-spline VP algorithm were applied. The results are similar to those in Tab.~\ref{tab:convtest_avg}, which are based on $100$ VP iterations. Note that our knot-prediction outperforms the uniform and random knot initialization procedures UVP and RVP, as their PRDNs are much higher. The lowest PRDN was achieved by the KR method for all recordings, but the average difference from our approach was less than $1.6\%$. Even though the KR algorithm yielded the best results in terms of PRDN, it was slow compared to our algorithm: The KR algorithm needed about $60$ minutes to predict the optimal knot positions, while the B-spline VP techniques required only $5$ minutes. For a whole $30$ minute recording, our knot-prediction was very fast for $p=2,\infty$, but for $p=1$ it took slightly longer (2 minutes). This is due to the computation of the median in Eq.~\eqref{eq:recall_lpcoeff}, which includes sorting in Alg.~\ref{alg:knotpred}. We conclude that applying our knot-prediction method followed by a few B-spline VP iterations provides good results. Compared to other algorithms, such as KR, our method was able to achieve a similar reconstruction error within reasonable time. Note that the sparse implementations of the matrices in Eq.~\eqref{eq:varpro_grad} have a great impact on execution time. We found that running the nonsparse B-spline VP method on the same test data took more than $10$ minutes. The CR value was the same for all algorithms because we considered B-spline approximations with the same parameters. Note that we computed the compression ratio at an algorithmic level only. Better values could be obtained by considering the CR at an implementational level. For instance, one of the key tasks is to find an adequate bit representation for storing the parameters of the B-spline approximations (i.e.\ quantization of coefficients and knots).


\section{Conclusion}
\label{sec:conclusion}
We have proposed an efficient algorithm for estimating the optimal position of free knots for nonlinear least-squares B-spline fitting. Our approach provides three different strategies for knot-prediction. It is based on the best first-order B-spline approximation in terms of $\ell_p$ norms for $p=1,2,\infty$. Application of these heuristics may depend on the specific task. For instance, $\ell_1$ norm solutions are more suitable for processing noisy data, while the uniform $\ell_\infty$ approximations are preferable for detecting particular structures in signals, such as spikes. The $\ell_2$ norm solution can be interpreted as a good tradeoff that is easy to compute. Another option would be to combine these solutions. For example, the role of the QRS complex can be emphasized by choosing the first few knots using the $\ell_\infty$ norm FOBA algorithm, followed by the $\ell_1$ or $\ell_2$ constraints. 

We also developed a sparse implementation for evaluating B-spline functions and the corresponding partial derivatives with respect to their free knots. A VP algorithm was adapted to refine our initial predictions. We have demonstrated the efficiency of this method using both synthetic and real-world data. We have also shown that the knot-prediction method along with the B-spline VP algorithm can be used for successful compression of real-world ECG recordings. The reconstructed signal has a simple analytic representation that can be used in further processing steps, such as smoothing, feature extraction, classification. 

Since its computational complexity is very low in the case of $\ell_2,\,\ell_\infty$ norms, the FOBA algorithm can also be employed to estimate the optimal number of knots. It is especially useful when no a-priori information about the signal is available. Note that we did not make use of specific properties of ECG data; thus the proposed method is applicable to various types of signals. 

\section{Acknowledgements}
P.\ Kov\'acs was supported by ELTE E\"otv\"os Lor\'and University within the \'UNKP-17-4 New National Excellence Program of the Ministry of Human Capacities. The project was supported by the Hungarian Scientific Research Fund (OTKA), project No K115804. The author would like to thank Cs. J. Heged\H{u}s for his useful remarks and comments. We are grateful to I. Abfalter for proofreading the manuscript.


\bibliographystyle{plain}
\bibliography{refs}


\end{document}